\documentclass[a4paper,11pt]{article} 
\usepackage[margin=1in]{geometry}

\usepackage{graphicx}
\usepackage{graphicx}
\usepackage{amsmath}
\usepackage{amsthm}
\usepackage{amssymb}
\usepackage{mathrsfs}
\usepackage{verbatim}
\usepackage{algorithm}
\usepackage{mathtools}
\usepackage{xspace}
\usepackage{enumitem}
\usepackage{color}
\usepackage{url}
\usepackage[all]{xy}
\usepackage{microtype}
\usepackage{array}
\usepackage{float}
\usepackage{algorithmic}
\usepackage{lmodern}
\usepackage{anyfontsize}
\usepackage[normalem]{ulem}
\usepackage{epigraph}
\usepackage{lineno}
\usepackage{dsfont}

\newcommand{\R}{\mathbb{R}}
\newcommand{\Z}{\mathbb{Z}}

\newcommand{\V}{\mathbb{V}}
\newcommand{\W}{\mathbb{W}}
\newcommand{\E}{\mathbb{E}}

\newcommand{\ri}{\mathcal{R}}
\newcommand{\rin}{\mathcal{R}^{\infty}}
\newcommand{\eps}{\varepsilon}
\newcommand{\Cech}{\v{C}ech\xspace}
\newcommand{\cech}{\mathcal{C}}

\newcommand{\ux}{\mathcal{X}}

\newcommand{\field}{\mathcal{F}}

\newcommand{\image}{\mathrm{im}}

\newcommand{\ch}{\mathcal{C}}

\newcommand{\simplicialmap}{\varphi}
\newcommand{\chainmap}{\phi}
\newcommand{\linearmap}{\lambda}

\newcommand{\ignore}[1]{}

\newcommand{\tg}{\tilde{g}}

\newcommand{\lin}{L_\infty}

\newcommand{\dmn}{diam_\infty}

\DeclarePairedDelimiter{\ceil}{\lceil}{\rceil}

\newtheorem{lemma}{Lemma}
\numberwithin{lemma}{section}
\newtheorem{definition}[lemma]{Definition}
\newtheorem{theorem}[lemma]{Theorem}

\begin{document}

\title{Improved Approximate Rips Filtrations with \\Shifted Integer Lattices
and Cubical Complexes
}

\author{
Aruni Choudhary\footnote{Institut f\"ur Informatik, Freie Universit\"at Berlin, Berlin, Germany \texttt{(aruni.choudhary@fu-berlin.de)}}
\and 
Michael Kerber\footnote{Graz University of Technology, 
Graz, Austria\texttt{(kerber@tugraz.at)}} 
\and
Sharath Raghvendra\footnote{Virginia Tech, 
Blacksburg, USA \texttt{(sharathr@vt.edu)}}
}

\maketitle

\begin{abstract}
Rips complexes are important structures for analyzing topological features of 
metric spaces. 
Unfortunately, generating these complexes is expensive 
because of a combinatorial explosion in the complex size.
For $n$ points in $\mathbb{R}^d$,
we present a scheme to construct a $2$-approximation of the
filtration of the Rips complex in the $L_\infty$-norm, 
which extends to a $2d^{0.25}$-approximation in the Euclidean case.
The $k$-skeleton of the resulting approximation 
has a total size of $n2^{O(d\log k +d)}$. 
The scheme is based on the integer lattice 
and simplicial complexes based on the barycentric subdivision of the $d$-cube.

We extend our result to use cubical complexes in place of simplicial complexes
by introducing \emph{cubical maps} between complexes.
We get the same approximation guarantee as the simplicial case,
while reducing the total size of the approximation 
to only $n2^{O(d)}$ (cubical) cells.

There are two novel techniques that we use in this paper.
The first is the use of \emph{acyclic carriers} for proving 
our approximation result.
In our application, these are maps which relate the Rips complex and the 
approximation in a relatively simple manner and greatly reduce the complexity 
of showing the approximation guarantee.
The second technique is what we refer to as \emph{scale balancing}, which is 
a simple trick to improve the approximation ratio under certain conditions.

\end{abstract}

\section{Introduction}
\label{section:intro}

\paragraph{Context.}

\emph{Persistent homology}~\cite{carlsson-survey,eh-book,elz-topological} is a 
technique to analyze data sets using topological invariants. 
The idea is to build a multi-scale representation of data sets
and to track its homological changes across the scales.

A standard construction for the important case of point clouds in
Euclidean space is the \emph{Vietoris-Rips complex} 
(usually abbreviated as simply the \emph{Rips complex}): 
for a scale parameter $\alpha\ge 0$, it is 
the collection of all subsets of points with diameter at most $\alpha$.
When $\alpha$ increases from $0$ to $\infty$, the Rips complexes form
a \emph{filtration}, an increasing sequence of nested simplicial complexes
whose homological changes can be computed and represented in terms of a 
\emph{barcode}.

The computational drawback of Rips complexes is their sheer size:
the $k$-skeleton of a Rips complex (that is, where only subsets of size 
at most $k+1$
are considered) for $n$ points consists of $\Theta(n^{k+1})$ simplices
because every $(k+1)$-subset joins the complex for a sufficiently large
scale parameter. 
This size bound makes barcode computations for large point clouds
infeasible even for low-dimensional homological features\footnote{An exception
are point clouds in $\R^2$ and $\R^3$, for which \emph{alpha complexes}~\cite{eh-book}
are an efficient alternative.}.
This difficulty motivates the question of what we can say
about the barcode of the Rips filtration 
without explicitly constructing all of its simplices.

We address this question using approximation techniques. 
The space of barcodes forms a metric space: 
two barcodes are close if similiar 
homological features occur on roughly the same range of scales.
More precisely, the bottleneck distance is used as a distance
metric between barcodes.
The first approximation scheme by~\cite{sheehy-rips}
constructs a $(1+\eps)$-approximation of the $k$-skeleton of the 
Rips filtration
using only $n(\frac{1}{\eps})^{O(\lambda k)}$ simplices for arbitrary finite
metric spaces, 
where $\lambda$ is the doubling dimension of the metric.
Further approximation techniques for Rips complexes~\cite{dfw-gic} 
and the closely related 
\emph{\Cech complexes}~\cite{bs-approximating,cjs-cech,ks-wssd} have 
been derived subsequently, all with comparable size bounds.
More recently, we constructed an approximation scheme~\cite{ckr-digital} for the 
\Cech filtrations of $n$ points in $\R^d$
that had size $n\left(\frac{1}{\eps}\right)^{O(d)}2^{O(d\log d +dk)}$
for the $k$-skeleton, improving the size bound from previous work.

In~\cite{ckr-polynomial-dcg}, we constructed an approximation scheme for 
Rips filtration in Euclidean space that yields a worse approximation factor 
of only $O(d)$, but uses only $n2^{O(d\log k +d)}$ simplices.
In~\cite{ckr-polynomial-dcg}, we also show a lower bound result
on the size of approximations: for any $\eps < 1/\log^{1+c} n$ with 
some constant $c\in (0,1)$, any $\eps$-approximate filtration has size 
$n^{\Omega(\log \log n)}$.

There has also been work on using cubical complexes to compute persistent
homology, such as in~\cite{wcv-cubical}.
Cubical complexes are typically smaller than their simplicial counterparts,
simply because they avoid triangulations. 
However, to our knowledge, there has been no attempt to utilize them
in computing approximations of filtrations.
Also, while there are efficient methods to compute persistence for simplicial
complexes connected with simplicial maps~\cite{dfw-gic,ks-twr}, we are not
aware of such counterparts for cubical complexes.

\paragraph{Our contributions.}

For the Rips filtration of $n$ points in $\R^d$ with distances taken in 
the $L_\infty$-norm, we present a $2$-approximation 
whose $k$-skeleton has size at most 
\[
n6^{d-1}(2k+4)(k+3)! \left\{\begin{array}{c}d\\k+2\end{array}\right\}
=n2^{O(d\log k + d)}
\]
where
$ \left\{\begin{array}{c}a\\b\end{array}\right\}$ denotes Stirling
numbers of the second kind.
This translates to a $2d^{0.25}$-approximation of the 
Rips filtration in the Euclidean metric and hence improves the asymptotic
approximation quality of our previous approach~\cite{ckr-polynomial-dcg} 
with the same size bound.
Our scheme gives the best size guarantee over all previous
approaches. 

On a high level, our approach follows a straightforward approximation scheme:
given a scaled and appropriately shifted integer grid on $\R^d$, 
we identify those grid points that are close to the input points and 
build an approximation complex using these grid points.
The challenge lies in how to connect these grid points to a simplicial complex
such that close-by grid points are connected,
while avoiding too many connections to keep the size small. 
Our approach first selects a set of \emph{active faces}
in the cubical complex defined over the grid, and defines the approximation
complex using the barycentric subdivision of this cubical complex.

We also describe an output-sensitive algorithm to compute our approximation.
By randomizing the aforementioned shifts of the grids, we obtain a
worst-case running time of $n2^{O(d)}\log\Delta+2^{O(d)}M$ in expectation,
where $\Delta$ is the \emph{spread} of the point set (that is, the ratio
of the diameter to the closest distance of two points)
and $M$ is the size of the approximation.

Additionally, this paper makes the following technical contributions:

\begin{itemize}
\item 

We follow the standard approach of defining a sequence
of approximation complexes
and establishing an \emph{interleaving} between the Rips filtration 
and the approximation.
We realize our interleaving using \emph{chain maps} connecting 
a Rips complex at scale $\alpha$ to an approximation complex at scale $c\alpha$, 
and vice versa, with $c\geq 1$ being the approximation factor. 
Previous approaches~\cite{ckr-polynomial-dcg,dfw-gic,sheehy-rips} used
\emph{simplicial maps} for the interleaving, which induce an elementary
form of chain maps and are therefore more restrictive.

The explicit construction of such maps can be a non-trivial task. 
The novelty of our approach is that we avoid this construction 
by the usage of \emph{acyclic carriers}~\cite{munkres}.
In short, carriers are maps that assign
subcomplexes to subcomplexes under some mild extra conditions.
While they are more flexible, they still certify the existence of 
suitable chain maps, as we exemplify in Section~\ref{section:preliminaries}.
We believe that this technique is of general interest for the construction
of approximations of cell complexes.

\item We exploit a simple trick that we call \emph{scale balancing} to 
improve the quality of approximation schemes. 
In short, if the aforementioned interleaving maps from and to the 
Rips filtration do not increase the scale parameter by the same amount, 
one can simply multiply the scale parameter of the approximation by a constant. 
Concretely, given maps
\[
\phi_\alpha:\ri_\alpha\rightarrow \mathcal{X}_\alpha\qquad 
\psi_\alpha:\mathcal{X}_\alpha\rightarrow \ri_{c\alpha}
\]
interleaving the Rips complex $\ri_\alpha$ and the approximation 
complex $\mathcal{X}_\alpha$,
we can define $\mathcal{X}'_\alpha:=\mathcal{X}_{\alpha/\sqrt{c}}$ and obtain maps
\[\phi'_\alpha:\ri_\alpha\rightarrow \mathcal{X}'_{\sqrt{c}\alpha}\qquad 
\psi_\alpha:\mathcal{X}'_\alpha\rightarrow \ri_{\sqrt{c}\alpha}\]
which improves the interleaving from $c$ to $\sqrt{c}$.
While it has been observed that the same trick can be used for improving
the worst-case distance between Rips 
and \Cech filtrations\footnote{Ulrich Bauer, private communication}, 
our work seems to be the first to make use of it in
the context of approximations.

\item We extend our approximation scheme to use cubical complexes
instead of simplicial complexes, thereby achieving a marked reduction 
in size complexity.
To connect the cubical complexes at different scales, 
we introduce the notion of \emph{cubical maps}, which is
a simple extension of simplicial maps to the cubical case.
While we do not know of an algorithm that can compute persistence for the case
of cubical complexes with cubical maps, we believe that this is a first
step towards advocating the use of cubical complexes as approximating
structures.

\end{itemize}

Our technique can be combined with dimension reduction techniques 
in the same way as in~\cite{ckr-polynomial-dcg}
(see Theorems 19, 21, and 22 therein), with improved logarithmic factors.
We state the main results in the paper, while omitting the 
technical details.

\paragraph{Updates from the conference version.}

An earlier version of this paper appeared at the 25th European Symposium
on Algorithms~\cite{ckr-barycentric}.
In that version, we achieved a $3\sqrt{2}$-approximation of the $L_\infty$
Rips filtration and correspondingly, a $3\sqrt{2}d^{0.25}$-approximation
of the $L_2$ case.
In this version, we improve the weak interleaving of~\cite{ckr-barycentric}
to a strong interleaving to get improved approximation factors.
We expand upon the details of scale balancing, among other proofs
that were missing from the conference version.
We add the case of cubical complexes in this version.

There is a subtle yet important distinction between the approximation complexes
used in the conference version and the current result.
In the conference version, our simplicial complex was built 
using only \textit{active} faces, while the current version
uses both \textit{active} and \textit{secondary} faces 
(please see Section~\ref{section:simplicial_scheme} for definitions).
This makes it easier to relate the simplicial and the cubical
complexes in the current version.
On the other hand the complexes are different, hence the associated proofs
have been adapted accordingly.

\paragraph{Outline.}

We start by explaining the relevant topological concepts in 
Section~\ref{section:preliminaries}.
We give details of the integer grids that we use in 
Section~\ref{section:shifted_grids}.
In Section~\ref{section:simplicial_scheme} we present our approximation 
scheme that uses the barycentric subdivision, and present the computational
aspects in Section~\ref{section:bary_compute}.
The extension to cubical complexes is presented in 
Section~\ref{section:cubical_scheme}.
We discuss practical aspects of our scheme
and conclude in Section~\ref{section:conclusion}.
Some details of the strong interleaving from 
Section~\ref{section:simplicial_scheme} are deferred to 
Appendix~\ref{subsection:appendix-strong}.

\section{Preliminaries}
\label{section:preliminaries}
We briefly review the essential topological concepts needed.
More details are available in standard references
\cite{bss-metrics,ccggo-proximity,eh-book,hatcher,munkres}.

\paragraph{Simplicial complexes.}

A \emph{simplicial complex} $K$ on a finite set of elements $S$ 
is a collection of subsets $\{\sigma\subseteq S\}$ called \emph{simplices} 
such that each subset $\tau\subset\sigma$ is also in $K$.
The dimension of a simplex $\sigma\in K$ is $k:=|\sigma|-1$, 
in which case $\sigma$ is called a \emph{$k$-simplex}.
A simplex $\tau$ is a \emph{sub-simplex} of $\sigma$ if $\tau\subseteq\sigma$. 
We remark that, commonly a sub-simplex is called a ``face'' of a simplex,
but we reserve the word ``face'' for a different structure.
For the same reason, we do not introduce the common notation of
of ``vertices'' and ``edges'' of simplicial complexes, but rather refer
to $0$- and $1$-simplices throughout.
The \emph{$k$-skeleton} of $K$ consists of
all simplices of $K$ whose dimension is at most $k$.
For instance, the $1$-skeleton of $K$ is a graph 
defined by its $0$-simplices and $1$-simplices.

Given a point set $P\subset\R^d$ and a real number $\alpha\ge 0$,
the \emph{(Vietoris-)Rips} complex on $P$ at scale $\alpha$ consists of all
simplices $\sigma=(p_0,\ldots,p_k)\subseteq P$ such that $diam(\sigma)\le \alpha$,
where $diam$ denotes the diameter.
In this work, we write $\ri_\alpha$ for the Rips complex at scale $2\alpha$ 
with the Euclidean metric,
and $\rin_\alpha$ when using the metric of the $L_\infty$-norm.
In either way, a Rips complex is an example of a \emph{flag complex},
which means that whenever a set $\{p_0,\ldots,p_k\}\subseteq P$ has the property
that every $1$-simplex $\{p_i,p_j\}$ is in the complex, then the
$k$-simplex $\{p_0,\ldots,p_k\}$ is also in the complex.

A related complex is the \emph{\Cech complex} of $P$ at scale $\alpha$, which consists
of simplices of $P$ for which the radius of the minimum enclosing ball
is at most $\alpha$.
We do not study \Cech complexes in this paper, but we mention them briefly
while showing a connection with the Rips complex later in this section. 

A simplicial complex $K'$ is a \emph{subcomplex} of $K$ if $K'\subseteq K$.
For instance, $\ri_{\alpha}$ is a subcomplex of $\ri_{\alpha'}$ for 
$0\le\alpha\leq\alpha'$. 
Let $L$ be a simplicial complex.
Let $\hat{\simplicialmap}$ be a map which assigns a vertex of $L$
to each vertex of $K$.
A \emph{simplicial map} is a map $\simplicialmap:K\rightarrow L$  
induced by a vertex map $\hat{\simplicialmap}$, such that 
for every simplex $\{p_0,\ldots,p_k\}$ in $K$, the set 
$\{\hat{\simplicialmap}(p_0),\ldots,\hat{\simplicialmap}(p_k)\}$
is a simplex of $L$. 
For $K'$ a subcomplex of $K$, the inclusion map $inc:K'\rightarrow K$
is an example of a simplicial map. 
A simplicial map is completely determined
by its action on the $0$-simplices of $K$.

\paragraph{Chain complexes.}

A \emph{chain complex} $\ch_\ast=(\ch_p,\partial_p)$ with $p\in\Z$ is a collection 
of abelian groups $\ch_p$ and homomorphisms $\partial_p:\ch_p\rightarrow \ch_{p-1}$ 
such that $\partial_{p-1}\circ\partial_{p}=0$.
A simplicial complex $K$ gives rise to a chain complex $\ch_\ast(K)$ 
for a fixed base field $\mathcal{F}$:
define $\ch_p$ for $p\geq 0$ as the set of formal linear combinations of 
$p$-simplices in $K$ over $\mathcal{F}$, and $\ch_{-1}:=\mathcal{F}$.
The boundary of a $k$-simplex with $k\geq 1$ is the (signed) sum of its sub-simplices
of co-dimension one\footnote{To avoid thinking about orientations, it is often assumed
	that $\mathcal{F}=\Z_2$ is the field with two elements.}; 
the boundary of a $0$-simplex is simply set to $1$.
The homomorphisms $\partial_p$ are then defined as the linear extensions
of this boundary operator.
Note that $\ch_\ast(K)$ is sometimes called \emph{augmented chain complex}
of $K$, where the augmentation refers to the addition of the non-trivial group $\ch_{-1}$.

A \emph{chain map} $\chainmap:\ch_\ast\rightarrow D_\ast$ between
chain complexes $\ch_\ast=(\ch_p,\partial_p)$ and $D_\ast=(D_p,\partial'_p)$
is a collection of group homomorphisms $\chainmap_p:\ch_p\rightarrow D_p$
such that $\chainmap_{p-1}\circ\partial_{p}=\partial'_{p}\circ\chainmap_{p}$.
For simplicial complexes $K$ and $L$, we call a chain map
$\chainmap:\ch_\ast(K)\rightarrow\ch_\ast(L)$ \emph{augmentation-preserving}
if $\chainmap_{-1}$ is the identity.
A simplicial map $\simplicialmap:K\rightarrow L$ between simplicial complexes induces 
an augmentation-preserving chain map $\bar{\simplicialmap}:\ch_\ast(K)\rightarrow\ch_\ast(L)$ 
between the corresponding chain complexes.
This construction is \emph{functorial}, meaning that for 
$\simplicialmap$ the identity function on a simplicial complex $K$, 
$\bar{\simplicialmap}$ is the identity function on $\ch_\ast(K)$,
and for composable simplicial maps $\simplicialmap,\simplicialmap'$,
we have that 
$\overline{\simplicialmap\circ\simplicialmap'}=
\bar{\simplicialmap}\circ\bar{\simplicialmap'}$.

\paragraph{Homology.}

The \emph{$p$-th homology group} $H_p(\ch_\ast)$ of a
chain complex is defined 
as $\mathrm{ker}\,\partial_p/\mathrm{im}\,\partial_{p+1}$.
The $p$-th homology group of a simplicial complex $K$, $H_p(K)$, is the 
$p$-th homology group of its induced chain complex $\ch_\ast(K)$.
Note that this definition is commonly referred to as \emph{reduced} homology,
but we ignore this distinction and consider reduced homology throughout.
$H_p(\ch_\ast)$ is an $\field$-vector space because we have 
chosen our base ring $\mathcal{F}$ as a field.
Intuitively, when the chain complex is generated from a simplicial complex, 
the dimension of the $p$-th homology group counts the number of 
$p$-dimensional holes in the complex.
We write $H(\ch_\ast)$ for the direct sum of all $H_p(\ch_\ast)$ for $p\geq 0$.

A chain map $\chainmap:\ch_\ast\rightarrow D_\ast$ 
induces a linear map $\chainmap^\ast: H(\ch_\ast)\rightarrow H(D_\ast)$
between the homology groups. 
Again, this construction is functorial, meaning that it maps identity maps to 
identity maps, and it is compatible with compositions.

\paragraph{Acyclic carriers.}

We call a simplicial complex $K$ \emph{acyclic}, 
if $K$ is connected and all 
homology groups $H_p(K)$ are trivial.
For simplicial complexes $K$ and $L$, 
an \emph{acyclic carrier} $\Phi$ is a map that assigns to 
each simplex $\sigma$ in $K$, a non-empty acyclic 
subcomplex $\Phi(\sigma)\subseteq L$, 
and whenever $\tau$ is a sub-simplex 
of $\sigma$, then $\Phi(\tau)\subseteq\Phi(\sigma)$.
We say that a chain $c\in\ch_p(K)$ is \emph{carried} by a subcomplex $K'$, 
if $c$ takes value $0$ except for $p$-simplices in $K'$.
A chain map $\chainmap:\ch_\ast(K)\rightarrow \ch_\ast(L)$ is 
\emph{carried by $\Phi$}, if for each simplex $\sigma\in K$,
$\chainmap(\sigma)$ is carried by $\Phi(\sigma)$.
We state the \emph{acyclic carrier theorem}~\cite[Thm 13.3]{munkres},
adapted to our notation:
\begin{theorem}
\label{theorem:acyclic_carrier}
Let $\Phi:K\rightarrow L$ be an acyclic carrier.
Then,
\begin{itemize}
\item There exists an augmentation-preserving chain map 
$\chainmap:\ch_\ast(K)\rightarrow \ch_\ast(L)$ 
carried by $\Phi$.

\item If two augmentation-preserving chain maps 
$\chainmap_1,\chainmap_2:\ch_\ast(K)\rightarrow \ch_\ast(L)$
are both carried by $\Phi$, then 
$\chainmap_1^\ast=\chainmap_2^\ast$.\footnote{In the language of~\cite{munkres}, this result is stated
as the existence of a \emph{chain homotopy} between $\phi_1$ and $\phi_2$.
As evident from~\cite[Theorem\,~12.4]{munkres}, this implies that
the induced linear maps are the same.
}
\end{itemize}
\end{theorem}
We remark that ``augmentation-preserving'' is crucial in the statement:
without it, the trivial chain map (that maps everything to $0$)
turns the first statement trivial and easily leads to a counter-example
for the second claim.

\paragraph{Filtrations and towers.}

Let $I\subseteq\R$ be a set of real values which we refer to as \emph{scales}.
A \emph{filtration} is a collection of simplicial complexes
$(K_\alpha)_{\alpha\in I}$ such that $K_\alpha\subseteq K_\alpha'$ 
for all $\alpha\leq\alpha'\in I$. 
For instance, $(\ri_\alpha)_{\alpha\geq 0}$ is a filtration 
which we call the \emph{Rips filtration}.
A \emph{(simplicial) tower} is a sequence $(K_\alpha)_{\alpha\in J}$ 
of simplicial complexes with $J$ being a discrete set 
(for instance $J=\{2^k\mid k\in\Z\}$), together with simplicial maps 
$\simplicialmap_\alpha:K_\alpha\rightarrow K_{\alpha'}$ between
complexes at consecutive scales.
For instance, the Rips filtration can be turned into a tower by restricting 
to a discrete range of scales,
and using the inclusion maps as $\simplicialmap$.
The approximation constructed in this paper will be another example of a tower.

We say that a simplex $\sigma$ is \emph{included} in the tower at scale $\alpha'$,
if $\sigma$ is not in the image of the map
$\simplicialmap_{\alpha}:K_\alpha\rightarrow K_{\alpha'}$,
where $\alpha$ is the scale preceding $\alpha'$ in the tower.
The \emph{size} of a tower is the number of simplices included over all scales.
If a tower arises from a filtration, its size is simply the size of the 
largest complex in the filtration (or infinite, if no such complex exists).
However, this is not true in general for
simplicial towers, because simplices can collapse in the tower and the size of the
complex at a given scale may not take into account the collapsed simplices
which were included at earlier scales in the tower.

\paragraph{Barcodes and Interleavings.}

A collection of vector spaces $(V_\alpha)_{\alpha\in I}$ connected with linear maps
$\linearmap_{\alpha_1,\alpha_2}:V_{\alpha_1}\rightarrow V_{\alpha_2}$ 
is called a \emph{persistence module}, if $\linearmap_{\alpha,\alpha}$ is the 
identity on $V_\alpha$ and 
$\linearmap_{\alpha_2,\alpha_3}\circ\linearmap_{\alpha_1,\alpha_2}
=\linearmap_{\alpha_1,\alpha_3}$ 
for all $\alpha_1\le\alpha_2\le\alpha_3\in I$ for the index set $I$.

We generate persistence modules using the previous concepts. 
Given a simplicial tower $(K_\alpha)_{\alpha\in I}$,
we generate a sequence of chain complexes $(\ch_\ast(K_\alpha))_{\alpha\in I}$.
By functoriality, the simplicial maps $\simplicialmap$ of the tower 
give rise to chain maps $\overline{\simplicialmap}$ 
between these chain complexes.
Using functoriality of homology, we obtain a sequence 
$(H(K_\alpha))_{\alpha\in I}$ of vector spaces with linear maps 
$\overline{\simplicialmap}^\ast$, forming a persistence module. 
The same construction applies to filtrations as a special case.

Persistence modules admit a decomposition into a collection of 
intervals of the form $[\alpha,\beta]$
(with $\alpha,\beta\in I$), called the \emph{barcode}, subject to certain
tameness conditions.
The barcode of a persistence module characterizes the module uniquely up to 
isomorphism.
If the persistence module is generated by a simplicial complex,
an interval $[\alpha,\beta]$ in the barcode corresponds 
to a homological feature (a ``hole'')
that comes into existence at complex $K_\alpha$ and persists until
it disappears at $K_\beta$. 

Two persistence modules $(V_\alpha)_{\alpha\in I}$ and 
$(W_\alpha)_{\alpha\in I}$ 
with linear maps $\phi_{\cdot,\cdot}$ and $\psi_{\cdot,\cdot}$ are said to be  
\emph{weakly (multiplicatively) $c$-interleaved}
with $c\geq 1$, if 
there exist linear maps $\gamma_\alpha:V_\alpha\rightarrow 
W_{c\alpha}$ and $\delta_\alpha:W_\alpha\rightarrow V_{c\alpha}$,
called \emph{interleaving maps},
such that the diagram
\begin{equation}
\label{diagram:weak_diag}
\xymatrix{
& \cdots\ar[r] & V_{\alpha c} \ar[rd]^{\gamma}\ar[rr]^{\phi} & & 
V_{\alpha c^3} \ar[r] &\cdots
\\
\cdots\ar[r] & W_{\alpha} \ar[rr]^{\psi}\ar[ru]^{\delta} & & 
W_{\alpha c^2}\ar[r] \ar[ru]^{\delta} &\cdots
\\
}
\end{equation}
commutes, that is, 
$\psi=\gamma \circ \delta $ and $\phi= \delta\circ \gamma $ 
for all $\{\dots,\alpha/c^2,\alpha/c,\alpha,c\alpha,\dots \}\in I$
(we have skipped the subscripts of the maps for readability). 
In such a case, the barcodes of the two modules are $3c$-approximations 
of each other in the sense of~\cite{ccggo-proximity}. 
We say that two towers are \emph{$c$-approximations} of each other
if their persistence modules are $c$-approximations. 

Under the more stringent conditions of \emph{strong interleaving}, 
the approximation ratio can be improved. 
Two persistence modules $(V_\alpha)_{\alpha\ge 0}$ and $(W_\alpha)_{\alpha\ge 0}$
with respective linear maps $\phi_{\cdot,\cdot}$ and $\psi_{\cdot,\cdot}$
are said to be \emph{(multiplicatively) strongly $c$-interleaved} 
if there exist a pair of families of 
linear maps $\gamma_\alpha:V_{\alpha}\rightarrow W_{c\alpha}$ and 
$\delta_\alpha:W_{\alpha}\rightarrow V_{c\alpha}$ for $c>0$,
such that Diagram~\eqref{diag:strong_diag} 
commutes for all $0\le \alpha\le \alpha'$ (the subscripts of the maps
are excluded for readability).
In such a case, the persistence barcodes of the two modules are said 
to be $c$-approximations of each other in the sense of~\cite{ccggo-proximity}.

\begin{equation}
\label{diag:strong_diag}
\xymatrix{
	V_{\frac{\alpha}{c}} \ar[rrr]^{\phi} \ar[rd]^{\gamma} & & & V_{c\alpha'}    &  & V_{c\alpha} \ar[r]^{\phi} & V_{c\alpha'} 
	\\
	& W_\alpha \ar[r]^{\psi} & W_{\alpha'} \ar[ru]^{\delta} &                 & W_\alpha \ar[r]^{\psi} \ar[ru]^{\delta} & W_{\alpha'} \ar[ru]^{\delta}
	\\ 
	& V_\alpha \ar[r]^{\phi} & V_{\alpha'} \ar[rd]^{\gamma} &                 & V_\alpha \ar[r]^{\phi} \ar[rd]^{\gamma} & V_{\alpha'} \ar[rd]^{\gamma} 
	\\ 
	W_{\frac{\alpha}{c}} \ar[rrr]^{\psi} \ar[ru]^{\delta} & & & 
	W_{c\alpha'}    &  & W_{c\alpha} \ar[r]^{\psi} & W_{c\alpha'}\\
}
\end{equation}

Finally, we mention a special case that relates
equivalent persistence modules~\cite{cz-computing,handbook}.
Two persistence modules 
$\V=(V_\alpha)_{\alpha\in I}$ and $\W=(W_\alpha)_{\alpha\in I}$
that are connected through linear maps $\phi,\psi$ respectively are 
isomorphic if there exists an isomorphism 
$f_\alpha:V_\alpha\rightarrow W_\alpha$ for
each $\alpha\in I$ for which the following diagram
commutes for all $\alpha\le \beta \in I$:
\begin{eqnarray}
\label{equation:persisistence_equiv}
\xymatrix{
	\dots \ar[r] & V_\alpha \ar[r]^{\phi} \ar[d]^{f_\alpha} 
	& V_\beta \ar[r]  \ar[d]^{f_\beta} & \dots 
	\\
	\dots \ar[r] & W_\alpha \ar[r]^{\psi} & W_\beta \ar[r] & \dots 
}
\end{eqnarray}
Isomorphic persistence modules have identical persistence barcodes.

\paragraph{Scale balancing.}

Let $\V=(V_\alpha)_{\alpha\in I}$ and $\W=(W_\alpha)_{\alpha\in I}$ be two
persistence modules with linear maps $f_v,f_w$, respectively.
Let there be linear maps $\phi:V_{\alpha/\eps_1}\rightarrow W_{\alpha}$
and $\psi:W_{\alpha}\rightarrow V_{\alpha\eps_2}$ for 
$1\le \eps_1,\eps_2$ such that all $\alpha,\alpha/\eps_1,\alpha\eps_2\in I$.
Suppose that the following diagram commutes, for all $\alpha\in I$.
\begin{equation}
\label{diag:balance_orig}
\xymatrix{
& \ldots\ar[r] & W_{\alpha } \ar[rd]^{\psi}\ar[rr]^{f_w} & & 
W_{\alpha \eps_1 \eps_2} \ar[r] &\ldots
\\
\ldots\ar[r] & V_{\alpha/\eps_1} \ar[rr]^{f_v}\ar[ru]^{\phi} & & 
V_{\alpha \eps_2}\ar[r] \ar[ru]^{\psi} &\ldots
\\
}
\end{equation} 
Let $\eps:=max(\eps_1,\eps_2)$.
Then, by replacing $\eps_1,\eps_2$ by $\eps$ in Diagram~\eqref{diag:balance_orig},
the diagram still commutes, so $\V$ is a $3\eps$-approximation of $\W$.

We define a new vector space $V'_{c \alpha}:=V_\alpha$, 
where $c=\sqrt{\frac{\eps_1}{\eps_2}}$ and $c\alpha\in I$.
This gives rise to a new persistence module, $\V'=(V_{c\alpha})_{\alpha\in I}$.
The maps $\phi$ and $\psi$ can then be interpreted as 
$\phi:V'_{\alpha/\sqrt{\eps_1\eps_2}}\rightarrow W_{\alpha}$, or
$\phi:V'_{\alpha}\rightarrow W_{\alpha\sqrt{\eps_1\eps_2}}$
and $\psi:W_{\alpha}\rightarrow V'_{\alpha\sqrt{\eps_1\eps_2}}$.
Then, Diagram~\eqref{diag:balance_orig} can be re-interpreted as 
\begin{equation}
\label{diag:balance_scaled}
\xymatrix{
& \ldots\ar[r] & W_{\alpha \sqrt{\eps_1\eps_2}} \ar[rd]^{\psi}\ar[rr]^{f_w} & & 
W_{\alpha (\eps_1 \eps_2)^{3/2}} \ar[r] &\ldots
\\
\ldots\ar[r] & V_{\alpha'} \ar[rr]^{f_v}\ar[ru]^{\phi} & & 
V'_{\alpha \eps_1\eps_2}\ar[r] \ar[ru]^{\psi} &\ldots
\\
}
\end{equation} 
which still commutes.
Therefore, $\V'$ is a $3\sqrt{\eps_1\eps_2}$-approximation of $\W$, which is an 
improvement over $\V$, since $\sqrt{\eps_1\eps_2}\le max(\eps_1,\eps_2)$.
$\V$ and $\V'$ have the same barcode up to a scaling factor.

This scaling trick also works when $\V$ and $\W$ are strongly interleaved.
If we have the following commutative diagrams: (where we have skipped the maps 
for readability):
\begin{equation}
\label{diag:balance_strong_orig}
\xymatrix@C-1.0pc{
	W_{\alpha} \ar[rrr] \ar[rd] & & & W_{\alpha'\eps_1\eps_2}    &  & W_{\alpha\eps_1} \ar[r] & W_{\alpha'\eps_1} 
	\\
	& V_{\alpha\eps_2} \ar[r] & V_{\alpha'\eps_2} \ar[ru] &  & V_{\alpha} \ar[r]\ar[ru] & V_{\alpha'} \ar[ru]
	\\ 
	& W_{\alpha\eps_1} \ar[r] & W_{\alpha'\eps_1} \ar[rd] & & W_\alpha \ar[r] \ar[rd] & W_{\alpha'} \ar[rd] 
	\\ 
	V_{\alpha} \ar[rrr] \ar[ru] & & & 
	V_{\alpha'\eps_1\eps_2}    &  & V_{\alpha\eps_2} \ar[r] & V_{\alpha'\eps_2}
	\\
}
\end{equation}
then $\V$ and $\W$ are 
$max(\eps_1,\eps_2)$-approximations of each other.
By defining $\V'$ as before, the following diagrams 
\begin{equation}
\label{diag:balance_strong_scaled}
\xymatrix@C-0.0pc@R-0.0pc{
	W_{\alpha} \ar[rrr] \ar[rd] & & & W_{\alpha d^2}    
	&  & W_{\alpha d} \ar[r] & W_{\alpha' d} 
	\\
	& V'_{\alpha d} \ar[r] & V'_{\alpha' d} \ar[ru] &  & 
	V'_{\alpha} \ar[r]\ar[ru] & V'_{\alpha'} \ar[ru]
	\\ 
	& W_{\alpha d} \ar[r] & W_{\alpha' d} \ar[rd] & & 
	W_\alpha \ar[r] \ar[rd] & W_{\alpha'} \ar[rd] 
	\\ 
	V'_{\alpha} \ar[rrr] \ar[ru] & & & 	V'_{\alpha' d^2}    &  & 
	V'_{\alpha d} \ar[r] & V'_{\alpha' d}
	\\
}
\end{equation}
commute for $d=c\eps_2=\sqrt{\eps_1\eps_2}$, so we can improve a 
$\max(\eps_1,\eps_2)$-approximation to an $\sqrt{\eps_1\eps_2}$-approximation.

We end the section by discussing a basic but important relation 
between \Cech and Rips filtrations.
It is well-known that for any $\alpha\ge 0$, 
$\cech_\alpha\subseteq \ri_{\alpha}\subseteq 
\cech_{\sqrt{2}\alpha}$~\cite{eh-book}.
This gives a strong interleaving between the towers 
$(\cech_\alpha)_{\alpha\ge 0}$ and $(\ri_\alpha)_{\alpha\ge 0}$
with $\eps_1=1$ and $\eps_2=\sqrt{2}$.
Applying the scale balancing technique, we get that
\begin{lemma}
\label{lemma:cech-rips-filt}
The scaled \Cech persistence module 
$(H(\cech_{\sqrt[4]{2}\alpha}))_{\alpha\ge 0}$ and 
the Rips persistence module 
$(H(\ri_{\alpha}))_{\alpha\ge 0}$ 
are $\sqrt[4]{2}$-approximations of each other.
\end{lemma}

\section{Shifted Integer Lattices}
\label{section:shifted_grids}

In this section, we take a look at simple modifications of the integer lattice.

We denote by $I:=\{\alpha_s:=\lambda 2^s\mid s\in\Z\}$ with $\lambda>0$, 
a discrete set of scales. 
For each scale in $I$, we define grids which are scaled and translated (shifted) 
versions of the integer lattice.

\begin{definition}[scaled and shifted grids]
\label{def:shifted_grids}
For each scale $\alpha_s\in I$, we define the \emph{scaled and shifted grid}
$G_{\alpha_s}$ inductively as:
\begin{itemize}
\item For $s=0$, $G_{\alpha_{s}}$ is simply the scaled integer grid $\lambda\Z^d$, where
each basis vector has been scaled by $\lambda$.

\item For $s\geq 0$, we choose an arbitrary point
$O_{\alpha_{s}}\in G_{\alpha_{s}}$ and define
\begin{eqnarray}
\label{equation:grids_pos}
G_{\alpha_{s+1}} = 2(G_{\alpha_{s}}-O_{\alpha_{s}})+O_{\alpha_{s}}+
\frac{\alpha_s}{2}(\pm 1,\ldots,\pm 1),
\end{eqnarray}
where the signs of the components of the last vector are chosen 
independently and uniformly at random 
(and the choice is independent for each $s$). 

\item For $s\leq 0$,
we define
\begin{eqnarray}
\label{equation:grids_neg}
G_{\alpha_{s-1}} =
\frac{1}{2}(G_{\alpha_{s}}-O_{\alpha_{s}})+O_{\alpha_{s}}+
\frac{\alpha_{s-1}}{2}(\pm 1,\ldots,\pm 1),
\end{eqnarray}
where the last vector is chosen as in the case of $s\ge 0 $.
\end{itemize}	
\end{definition}

Equation~$\eqref{equation:grids_pos}$ 
and Equation~\eqref{equation:grids_neg} are consistent at $s=0$.
A simple example of the above construction is the sequence of grids
with $G_{\alpha_{s}}:=\alpha_s\Z^d$ for even $s$, and 
$G_{\alpha_{s}}:=\alpha_s\Z^d + \frac{\alpha_{s-1}}{2}(1,\ldots,1)$
for odd $s$. 

Next, we motivate the shifting of the grids. 
Let $\mathrm{Vor}_{G_s}(x)$ denote the Voronoi cell of any point $x\in G_s$
with respect to the point set $G_s$.
It is clear that the Voronoi cell 
is a cube of side length $\alpha_s$ centered at $x$. 
The shifting of the grids ensures that
each $x\in G_{\alpha_{s}}$ lies in the Voronoi region of a unique 
$y\in G_{\alpha_{s+1}}$.
Using an elementary calculation, we show a stronger statement:

\begin{lemma}
\label{lemma:vorcontain}
Let $x\in G_{\alpha_{s}}, y\in G_{\alpha_{s+1}}$ be such that $x\in \mathrm{Vor}_{G_{\alpha_{s+1}}}(y)$. 
Then, 
\[
\mathrm{Vor}_{G_{\alpha_{s}}}(x)\subset \mathrm{Vor}_{G_{\alpha_{s+1}}}(y).
\]
\end{lemma}

\begin{proof}
Without loss of generality, we can assume that $\alpha_s=2$ and $x$ is the 
origin, using an appropriate translation and scaling. 
Also, we assume for the sake of simplicity that 
$G_{\alpha_{s+1}}=2G_{\alpha_{s}} + (1,\ldots,1)$;
the proof is analogous for any other translation vector.
In that case, it is clear that $y=(1,\ldots,1)$. 
Since $G_{\alpha_{s}}=2\Z^d$, the Voronoi region of $x$ is the set $[-1,1]^d$.
Since $G_{\alpha_{s+1}}$ is a translated version of $4\Z^d$,
the Voronoi region of $y$ is the cube $[-1,3]^d$, which covers $[-1,1]^d$.
The claim follows.
For an example look to Figure~\ref{figure:nested_grids}.
\end{proof}

\begin{figure}[h]
\centering
\includegraphics[width=0.6\textwidth]{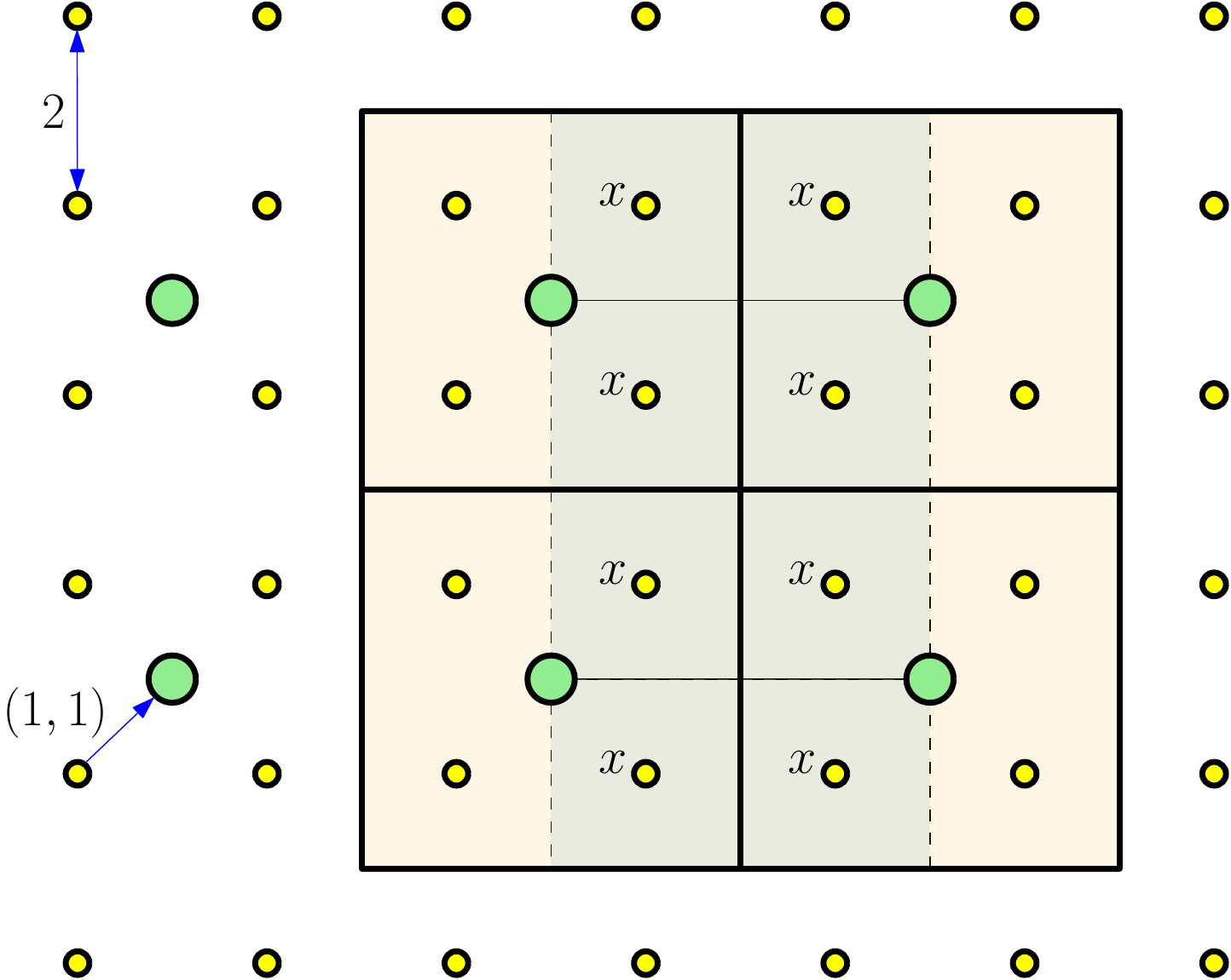}
\caption{$G_{\alpha_{s}}$ is represented by small disks (yellow), 
while $G_{\alpha_{s+1}}$ is represented by larger disks (green).
Possible locations of $x$ are indicated with their Voronoi regions. 
The Voronoi regions of the larged grid contain those of $x$.}
\label{figure:nested_grids}
\end{figure}

\subsection{Cubical complex of $\Z^d$}
\label{subsection:cubical-basic}

The integer grid $\Z^d$ naturally defines a cubical complex,
where each element is an axis-aligned, $k$-dimensional cube with $0\leq k\leq d$.
To define it formally,
let $\square$ denote the set of all integer translates
of faces of the unit cube $[0,1]^d$, considered as a convex polytope in $\R^d$.
We call the elements of $\square$ \emph{faces of $\Z^d$}. 

Each face has a dimension $k$; the $0$-faces,
or \emph{vertices} are exactly the points in $\Z^d$. 
The \emph{facets} of a $k$-face $E$ are the $(k-1)$-faces contained in $E$.
We call a pair of facets of $E$ \emph{opposite facets},
if they are disjoint.
Naturally, these concepts carry over to scaled and shifted
versions of $\Z^d$, so we define $\square_{\alpha_{s}}$ as the cubical complex
defined by $G_{\alpha_{s}}$.

We define a map 
$g_{\alpha_{s}}: \square_{\alpha_{s}}\rightarrow \square_{\alpha_{s+1}}$ 
as follows:
for vertices of $\square_{\alpha_{s}}$, we assign to $x\in G_{\alpha_{s}}$ 
the (unique) vertex $y\in G_{\alpha_{s+1}}$ such that
$x\in \mathrm{Vor}_{G_{\alpha_{s+1}}}(y)$ (see Lemma~\ref{lemma:vorcontain}).
For a $k$-face $f$ of $\square_{\alpha_{s}}$ with vertices 
$(p_1,\ldots,p_{2^k})$ in $G_{\alpha_{s}}$,
we set $g_{\alpha_{s}}(f)$ to be the convex hull of 
$\{g_{\alpha_{s}}(p_1),\ldots,g_{\alpha_{s}}(p_{2^k})\}$;
the next lemma shows that this is a well-defined map.
In this paper, we sometimes call $g_{\alpha_{s}}$ a \emph{cubical map}, 
since it is a counterpart of simplicial maps for cubical complexes.

\begin{lemma}
\label{lemma:gcell}
Let $f$ be $k$-face of $\square_{\alpha_{s}}$ with vertices 
$\{ p_1,\ldots,p_{2^k} \}\subset G_{\alpha_{s}}$.
Then 
\begin{itemize}
\item the set of vertices
$\{g_{\alpha_{s}}(p_1),\ldots,g_{\alpha_{s}}(p_{2^k})\}$ 
form a face $e$ of $\square_{\alpha_{s+1}}$.

\item for every face $e_1 \subset e$, there is a face $f_1 \subset f$
such that $g_{\alpha_s}(f_1)=e_1$. 

\item if $e_1,e_2$ are any two opposite facets of $e$, then there exists
a pair of opposite facets $f_1,f_2$ of $f$ such that $g_{\alpha_{s}}(f_1)=e_1$ 
and $g_{\alpha_{s}}(f_2)=e_2$.
\end{itemize} 
\end{lemma}

\begin{proof}
\textbf{First claim:}
We prove the first claim by induction on the dimension of faces of 
$G_{\alpha_{s}}$.
Base case: for vertices, the claim is trivial using 
Lemma~\ref{lemma:vorcontain}. 
Induction case: let the claim hold true for all $(k-1)$-faces of 
$G_{\alpha_{s}}$.
We show that the claim holds true for all $k$-faces of $G_{\alpha_{s}}$.

Let $f$ be a $k$-face of $G_{\alpha_{s}}$. 
Let $f_1$ and $f_2$ be opposite facets of $f$, along the $m$-th coordinate.
Let us denote the vertices of $f_1$ by $(p_1,\ldots,p_{2^{k-1}})$
and those of $f_2$ by $(p_{2^{k-1}+1},\ldots,p_{2^{k}})$ taken in the same order, that is,
$p_j$ and $p_{2^{k-1}+j}$ differ in only the $m$-th coordinate for all $1\le j\le 2^{k-1}$. 
By definition, all vertices of $f_1$ share the $m$-th coordinate, and we denote
coordinate of these vertices by $z$.
Then, the $m$-th coordinate of all vertices of $f_2$ equals $z+\alpha_s$.
Then $g_{\alpha_{s}}(p_j)$ and $g_{\alpha_{s}}(p_{2^{k-1}+j})$ have 
the same coordinates, except possibly the $m$-th coordinate. 
By induction hypothesis, $e_1=g_{\alpha_{s}}(f_1)$ and 
$e_2=g_{\alpha_{s}}(f_2)$ are two faces of $G_{s+1}$. 
This implies that $e_2$ is a translate of $e_1$ along the $m$-th
coordinate.

There are two cases: if $e_1$ and $e_2$ share the $m$-th coordinate, then
$e_1=e_2$ and therefore $g_{\alpha_{s}}(f)=e_1=e_2=e$, so the claim follows.
On the other hand, if $e_1$ and $e_2$ 
do not share the $m$-th coordinate, 
then they are two faces of $\square_{\alpha_{s+1}}$ 
which differ in only one coordinate by $\alpha_{s+1}$.
So they are opposite facets of a co-dimension one face $e$ of $G_{\alpha_{s+1}}$.
Using induction, the claim follows.

\textbf{Second claim:}
We prove the claim by induction over the dimension of $e_1$.
Base case: $e_1$ is a vertex.
The vertices of $f$ in Voronoi region of $e_1$ form $f_1$.
Since $f$ is an axis parallel face and the Voronoi region is also 
axis-parallel, it is immediate that $f_1$ is a face of $f$.
Assume that the claim is true up to dimension $i$.
For $e_1$ a face of dimension $i+1$, consider opposite facets $e_a$
and $e_b$ of $e$.
By the induction claim, there exist faces $f_a,f_b\subset f$
that satisfy $g_{\alpha_s}(f_a)=e_a, g_{\alpha_s}(f_b)=e_b$.
$f_a$ and $f_b$ are disjoint since otherwise $g_{\alpha_s}(f_a\cap f_b)$
would be common to both $e_a$ and $e_b$, a contradiction.
If $e_a$ is a translate of $e_b$ along the $m$-th coordinate,
then $f_a$ is also a translate of $f_b$ along the same coordinate.
Therefore $f_a$ and $f_b$ are opposite faces of a face $f_1$
and $g_{\alpha_s}(f_1)=e_1$.

\textbf{Third claim:} 
Without loss of generality, assume that $x_1$  is the direction
in which $e_2$ is a translate of $e_1$.
Using the second claim, 
let $h$ denote the maximal face of $f$ such that $g_{\alpha_{s}}(h)=e_1$.
Clearly, $h\neq f$, since that would imply $g_{\alpha_{s}}(f)=e_1=e$, 
which is a contradiction. 

Suppose $h$ has dimension less than $k-1$.
Let $h'$ be the facet of $f$ that contains $h$ and 
has the same $x_1$ coordinates for all vertices.
Then $g_{\alpha_s}(h')=e_1$, which contradicts the maximality
of $h$.

Therefore, the only possibility is that $h$ is a facet $f_1$ of $f$ such 
that $g_{\alpha_{s}}(f_1)=e_1$. 
Let $f_2$ be the opposite facet of $f_1$. 
From the proof of the first claim, it is easy to see that 
$g_{\alpha_{s}}(f_2)=e_2$. 
The claim follows.
\end{proof}

\subsection{Barycentric subdivision}

We discuss a special triangulation of $\square_{\alpha_{s}}$.
A \emph{flag} in $\square_{\alpha_{s}}$ is a set of faces $\{f_0,\ldots,f_k\}$ 
of $\square_{\alpha_{s}}$ such that 
\[
f_0\subseteq \ldots\subseteq f_k.
\]
The \emph{barycentric subdivision} of $\square_{\alpha_{s}}$, 
denoted by $sd_{\alpha_{s}}$, is the (infinite) simplicial complex whose 
simplices are the flags of $\square_{\alpha_{s}}$~\cite{munkres}.

In particular, the $0$-simplices of $sd_{\alpha_{s}}$ are the faces of 
$\square_{\alpha_{s}}$. 
An equivalent geometric description of $sd_{\alpha_{s}}$ can be obtained 
by defining the $0$-simplices as the barycenters of the faces in 
$sd_{\alpha_{s}}$, and introducing a $k$-simplex between $(k+1)$ barycenters 
if the corresponding faces form a flag.
For a simple example, see Figure~\ref{figure:barycentric_1} and Figure~\ref{figure:barycentric_2}.
It is easy to see that $sd_{\alpha_{s}}$ is a flag complex.
Given a face $f$ in $\square_{\alpha_{s}}$, we write $sd(f)$ for the 
subcomplex of $sd_{\alpha_{s}}$
consisting of all flags that are formed only by faces contained in $f$.
\begin{figure}[ht]
\centering
\includegraphics[width=0.4\columnwidth]{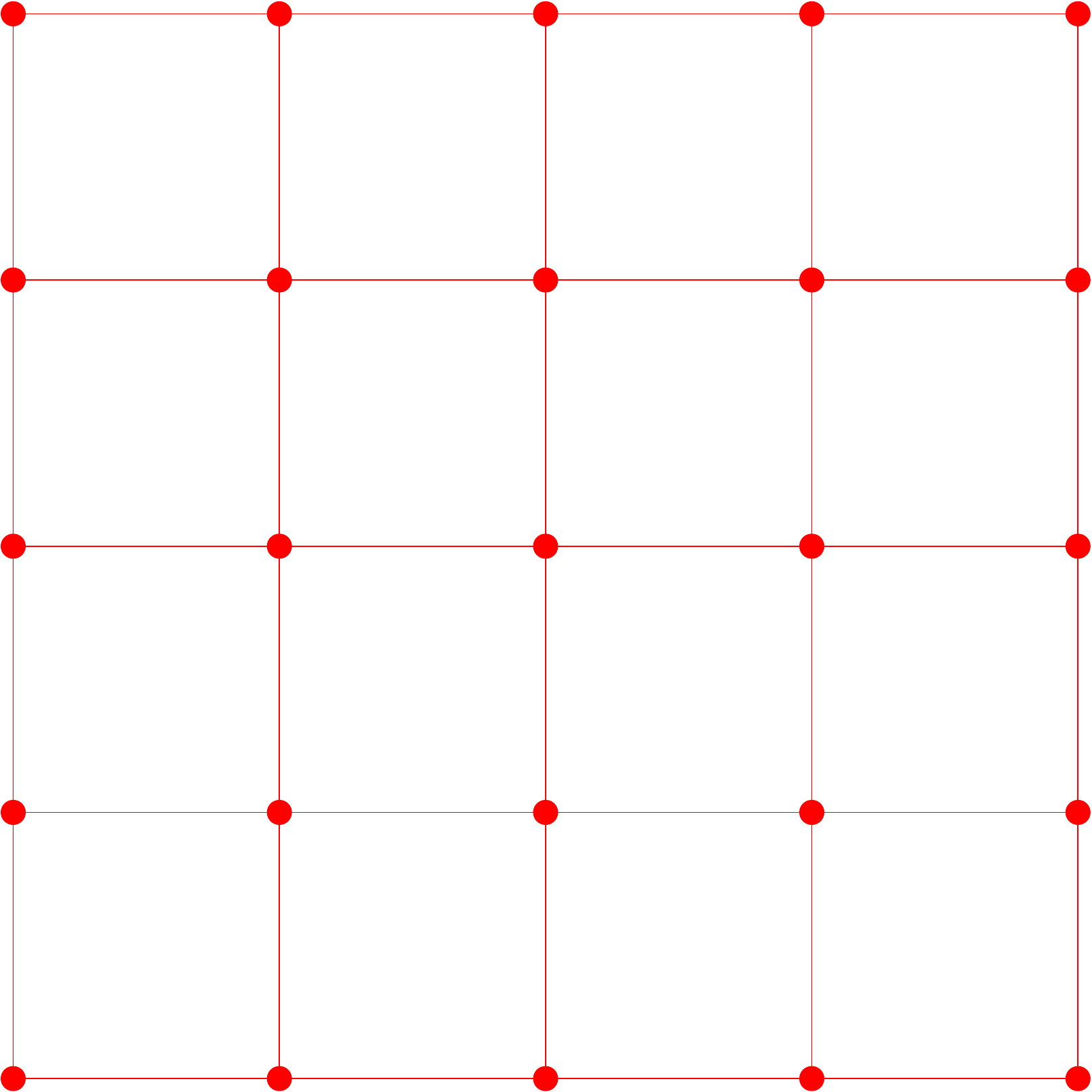}
\caption{A portion of the grid in two dimensions. 
The dots are the grid points which form the $0$-faces of the cubical complex.}
\label{figure:barycentric_1}
\end{figure}

\begin{figure}[ht]
\centering
\includegraphics[width=0.4\columnwidth]{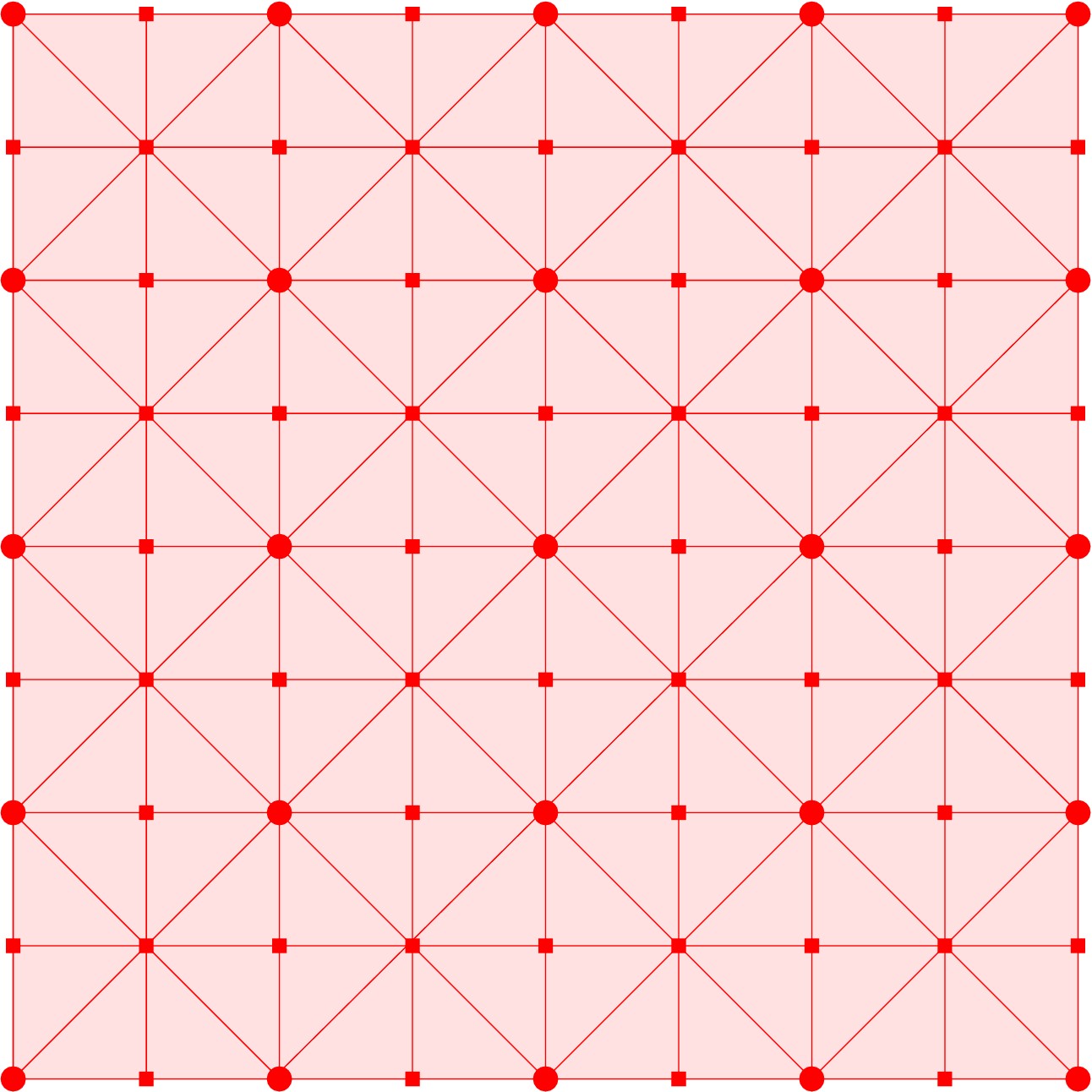}
\caption{The barycentric subdivision of the grid.
The tiny squares are barycenters of the $1$-faces and $2$-faces of the 
cubical complex. }
\label{figure:barycentric_2}
\end{figure}

\section{Approximation scheme with simplicial complexes}
\label{section:simplicial_scheme}

We define our approximation complex for a finite set of points in $\R^d$.
Recall from Definition~\ref{def:shifted_grids} that we can define a collection
of scaled and shifted integer grids $G_{\alpha_s}$ over a collection of scales
$I:=\{\alpha_s= 2^s \mid s\in \Z \}$ in $\R^d$.
To make the exposition simple, we define our complex in a slightly
generalized form.

\subsection{Barycentric spans}
\label{subsection:bary_span}

Fix some $s\in \Z$ and let $V$ denote any non-empty subset of $G_{\alpha_{s}}$.

\paragraph{Vertex span.}
\label{def:bary_vertex_span}
We say that a face $f\in\square_{\alpha_{s}}$ is \emph{spanned} by $V$,
if the set of vertices $V(f):=f\cap V$ 
\begin{itemize}
\item is non-empty, and

\item  not contained in any facet of $f$.
\end{itemize}

Trivially, the vertices of $\square_{\alpha_{s}}$ which are spanned by $V$ are precisely 
the points in $V$.
Any face of $\square_{\alpha_{s}}$ which is not a vertex must contain at least 
two vertices of $V$ in order to be spanned.
We point out that the set of spanned faces of $\square_{\alpha_{s}}$ 
is \emph{not} closed under taking sub-faces. 
For instance, if $V$ consists of two antipodal
points of a $d$-cube, the only faces spanned by $V$ are the $d$-cube
and the two vertices; all other faces of the $d$-cube contain at most one
vertex and hence are not spanned.

It is simple to test whether any given $k$-face $f\in \square_{\alpha_{s}}$ 
is spanned by the set of points $V(f)$. 
Let $T\subseteq [1,\ldots,d]$ be the set of common coordinates of the 
points in $V(f)$.
$V(f)$ spans $f$ if and only if the standard basis vectors of $\R^d$ 
corresponding to $T$ span $f$.
$T$ can be computed in $|V(f)|O(d)=O(2^{k}d)$ time by a linear scan
of the coordinates.
The coordinate directions spanned by $f$ can also be found and compared
with $T$ within the same time bound.

\paragraph{Barycentric span.}

The \emph{barycentric span} of $V$ is the subcomplex of $sd_{\alpha_{s}}$
obtained by taking the union of the complete barycentric subdivisions of the maximal 
faces of $\square_{\alpha_{s}}$ that are spanned by $V$.
The barycentric span of $V$ is indeed a simplicial complex by definition. 
Moreover, the barycentric span is a flag complex.
Then for any face $f\in \square_{\alpha_s}$, 
the barycentric span of $V(f)$ is either empty or acyclic.

Furthermore, for any non-empty subset $W\subseteq V$, 
the faces of $\square_{\alpha_{s}}$ that are spanned by $W$ are 
also spanned by $V$. 
Consequently, the barycentric span of $W$ is a subcomplex 
of the barycentric span of $V$.

\subsection{Approximation complex}
\label{subsection:bary_appcpx}

We denote by $P\subset \R^d$ a finite set of points.
We define two maps:
\begin{itemize}
\item $a_{\alpha_{s}}:P\rightarrow G_{\alpha_{s}}$: for each point $p\in P$, 
we let $a_{\alpha_{s}}(p)$ denote the grid point in $G_{\alpha_{s}}$ that is closest to $p$,
that is, $p\in \mathrm{Vor}_{G_{\alpha_{s}}}(a_{\alpha_{s}}(p))$. 
We assume for simplicity that this closest point is unique, which 
can be ensured using well-known methods~\cite{em-sos}. 
We define the \emph{active vertices of $G_{\alpha_s}$} as
\[
V_{\alpha_{s}}:=\image(a_{\alpha_{s}})=a_{\alpha_{s}}(P)\subset G_{\alpha_{s}},
\] 
that is, the set of grid points that have at least one point of $P$ 
in their Voronoi cells.

\item $b_{\alpha_{s}}:V_{\alpha_{s}}\rightarrow P$: the map $b_{\alpha_{s}}$ 
takes an active vertex of $G_{\alpha_{s}}$ to its closest point in $P$.
By taking an arbitrary total order on $P$ to resolve multiple assignments, 
we ensure that this assignment is unique.
\end{itemize}
Naturally, $b_{\alpha_{s}}(v)$ is a point inside $\mathrm{Vor}_{G_{\alpha_{s}}}(v)$ 
for any $v\in V_{\alpha_{s}}$.
It follows that the map $b_{\alpha_{s}}$ is a section 
of $a_{\alpha_{s}}$, 
that is, $a_{\alpha_{s}}\circ b_{\alpha_{s}}:V_{\alpha_{s}} \rightarrow V_{\alpha_{s}}$ is the identity on $V_{\alpha_s}$.
However, this is not true for $b_{\alpha_{s}}\circ a_{\alpha_{s}}$ in general.

Recall that the map 
$g_{\alpha_{s}}:\square_{\alpha_{s}}\rightarrow\square_{\alpha_{s+1}}$ 
takes grid points of $G_{\alpha_{s}}$ to grid points of $G_{\alpha_{s+1}}$. 
Using Lemma~\ref{lemma:vorcontain}, it follows at once that:
\begin{lemma}
\label{lemma:bary_gcompose}
For all $\alpha_{s}\in I$ and each $x\in V_{\alpha_{s}}$,
$g_{\alpha_{s}}(x)=(a_{\alpha_{s+1}}\circ b_{\alpha_{s}})(x)$.
\end{lemma}

Recall that $\rin_\alpha$ denotes the 
Rips complex at scale $\alpha$ for the $\lin$-norm.
The next statement is a direct application of the the triangle inequality; 
let $\dmn()$ denote the diameter in the $\lin$-norm. 
\begin{lemma}
\label{lemma:bary_iripscell}
Let $Q\subseteq P$ be a non-empty subset such that $\dmn(Q)\le \alpha_s$. 
Then, the set of grid points $a_{\alpha_{s}}(Q)$ is contained in a face 
of $\square_{\alpha_{s}}$. 

Equivalently, for any simplex 
$\sigma=(p_0,\ldots,p_k)\in\rin_{\alpha_s/2}$ on $P$, 
the set of active vertices $\{a_{\alpha_{s}}(p_0),\ldots,a_{\alpha_{s}}(p_k)\}$ 
is contained in a face of $\square_{\alpha_{s}}$.
\end{lemma}

\begin{proof}
We prove the claim by contradiction. 
Suppose that the set of active vertices $a_{\alpha_{s}}(Q)$ is not 
contained in a face of $\square_{\alpha_{s}}$. 
Then, there exists at least one pair of points $\{x,y\}\in Q$ such that
$a_{\alpha_{s}}(x)$, $a_{\alpha_{s}}(y)$ are not in a common face of $\square_{\alpha_{s}}$. 
By the definition of the grid $G_{\alpha_{s}}$, the grid points $a_{\alpha_{s}}(x)$, $a_{\alpha_{s}}(y)$ 
therefore have $L_\infty$-distance at least $2\alpha_s$.
Moreover, $x$ has $L_\infty$-distance less than $\alpha_s/2$ from $a_{\alpha_{s}}(x)$, 
and the same is true for $y$ and $a_{\alpha_{s}}(y)$.
By the triangle inequality, the $L_\infty$-distance of $x$ and $y$ 
is more than $\alpha_{s}$, which is a contradiction to the fact that 
$\dmn(Q)\le \alpha_{s}$.
\end{proof}

We now define our approximation tower.
For any scale $\alpha_{s}$, we define $\ux_{\alpha_{s}}$ as 
the barycentric span of the active vertices 
$V_{\alpha_{s}}\subset G_{\alpha_{s}}$.
See Figure~\ref{figure:bary_cpx_1}, Figure~\ref{figure:bary_cpx_2} and Figure~\ref{figure:bary_cpx_3} for a simple illustration. 
\begin{figure}[h]
\centering
\includegraphics[width=0.5\textwidth,page=2]{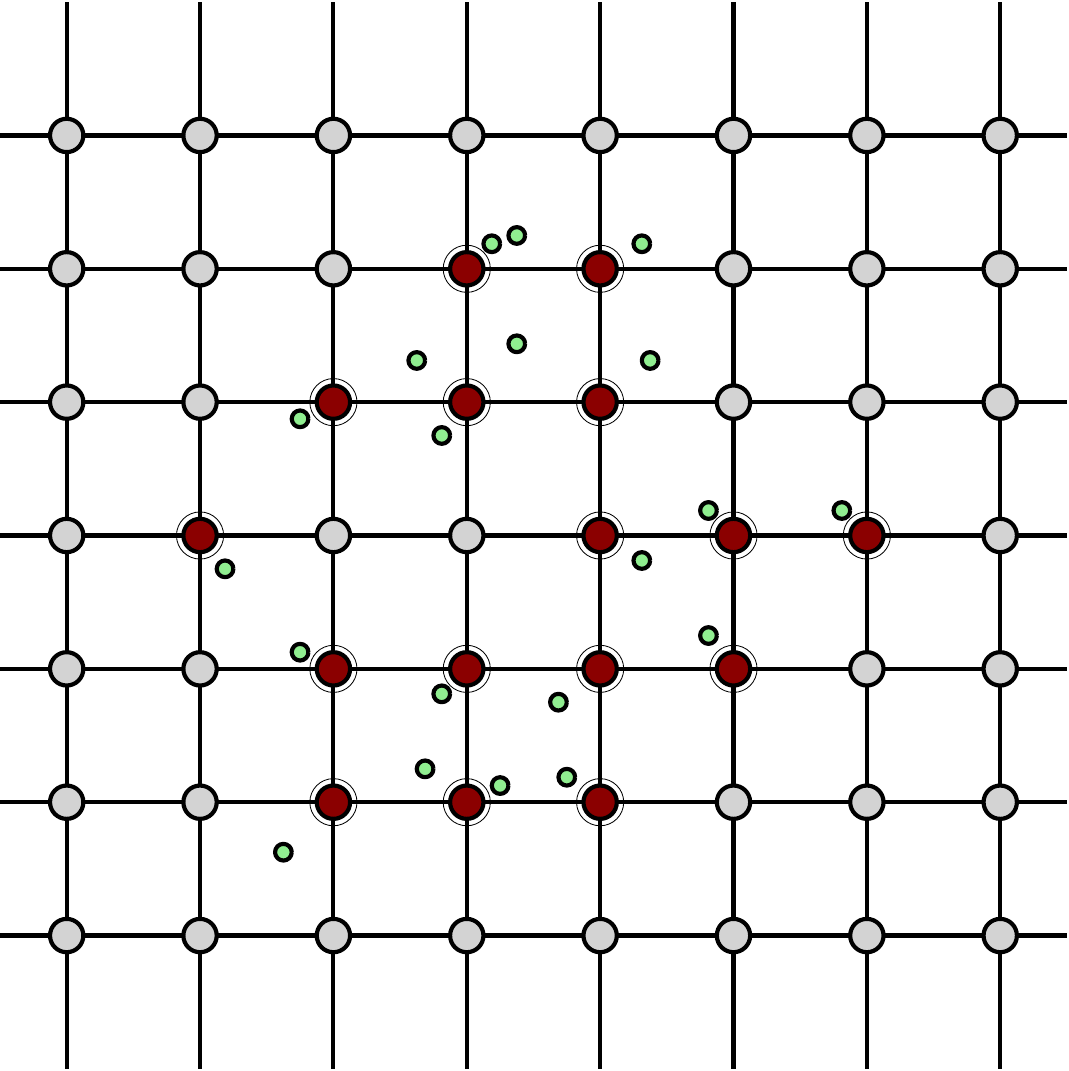}
\caption{A two-dimensional grid, shown along with its cubical complex.
The green points (small dots) denote the
points in $P$ and the red vertices (encircled) are the active vertices.}
\label{figure:bary_cpx_1}
\end{figure}
\begin{figure}[!h]
\centering
\includegraphics[width=0.5\textwidth,page=3]{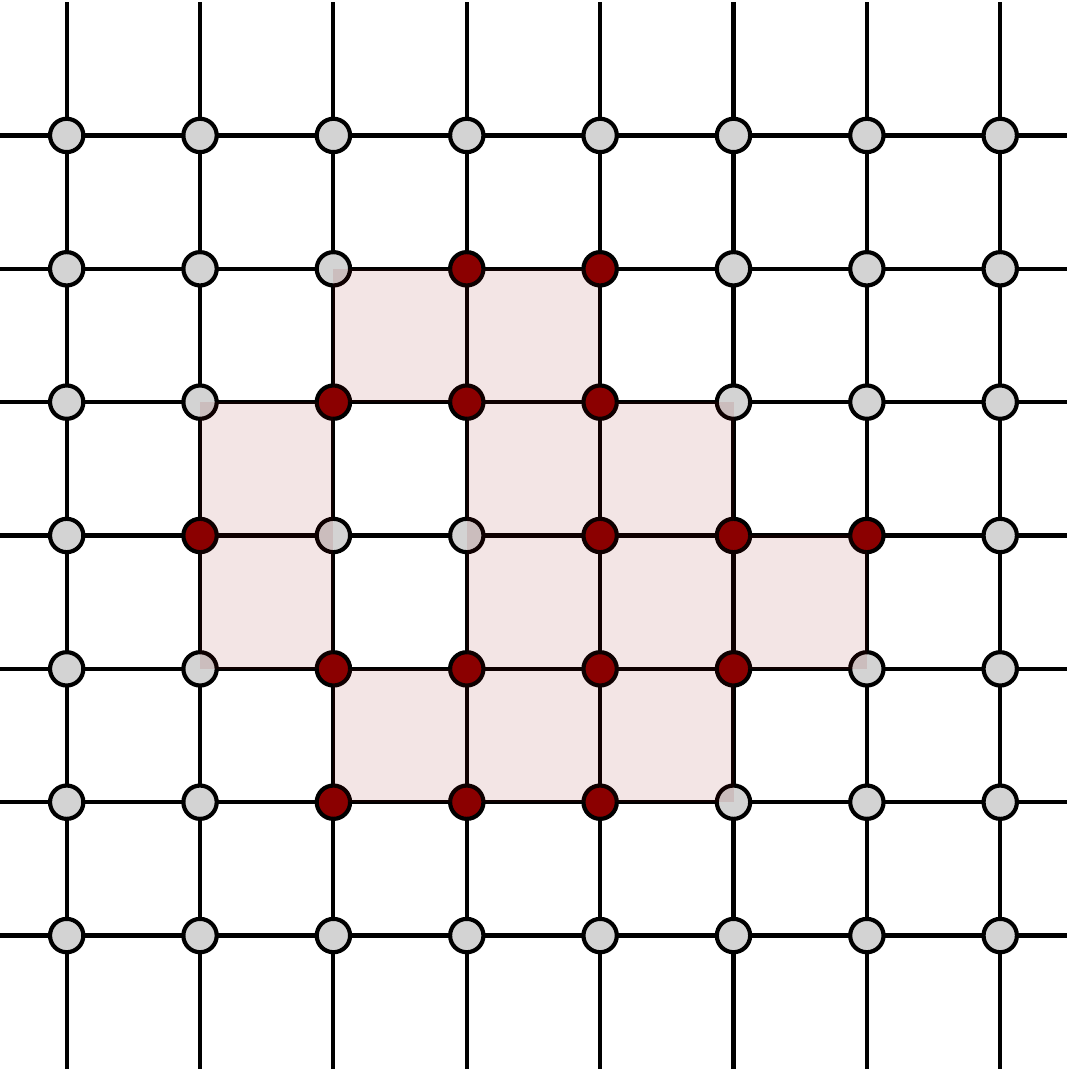}
\caption{The active faces are shaded. The closure of the active faces forms the cubical complex.}
\label{figure:bary_cpx_2}
\end{figure}
\begin{figure}
\centering
\includegraphics[width=0.5\textwidth,page=5]{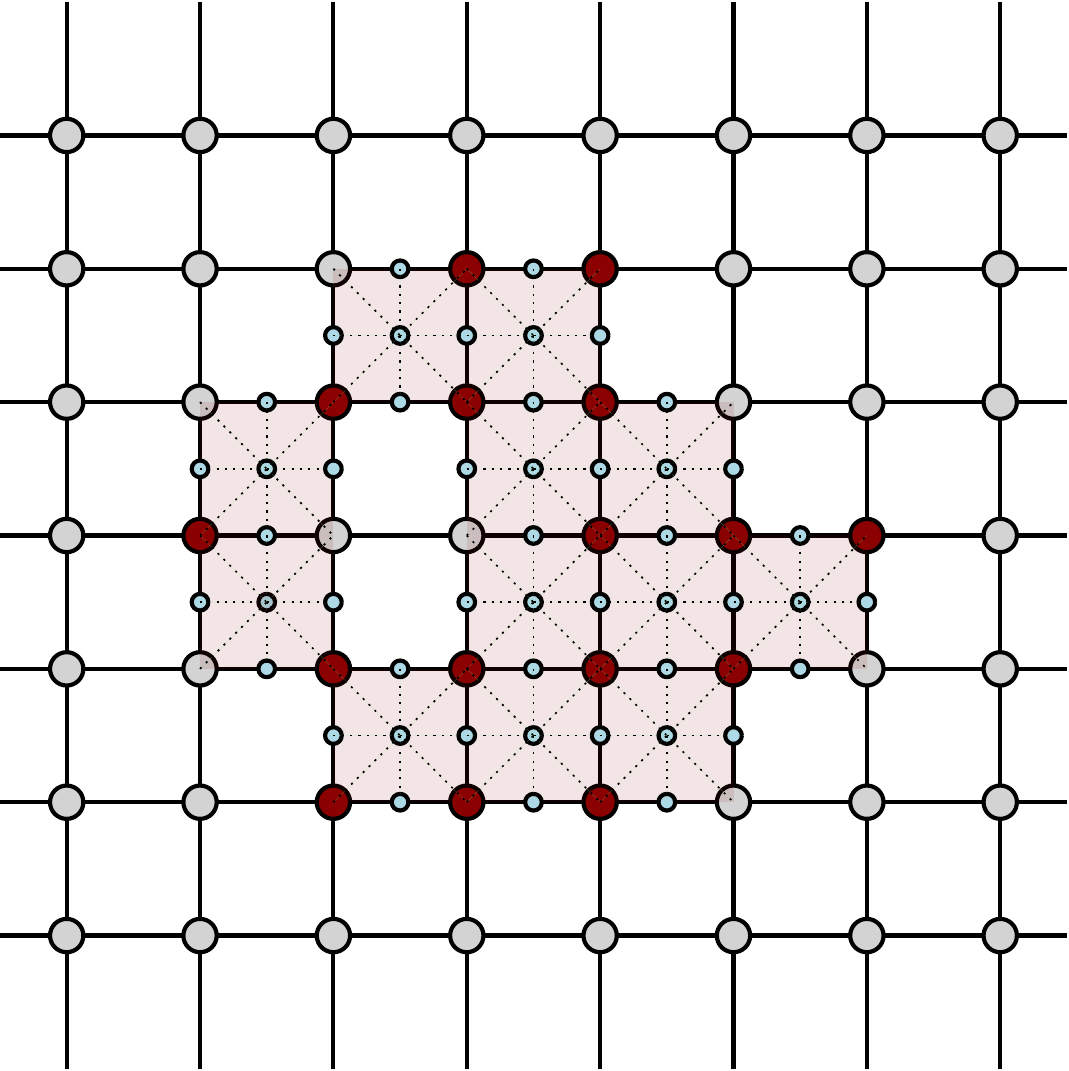}
\caption{The generated approximation complex, 
whose vertices consist of those of the cubical complex and
the blue vertices (small dots), which 
are the barycenters of active and secondary faces.}
\label{figure:bary_cpx_3}
\end{figure}

To simplify notation, we denote 
\begin{itemize}
\item the faces of $\square_{\alpha_{s}}$
spanned by $V_{\alpha_{s}}$ as \emph{active faces}, and

\item the faces of active faces that are not
spanned by $V_{\alpha_{s}}$ as \emph{secondary faces}.
\end{itemize}
To complete the description of the approximation tower, 
we need to define simplicial maps of the form
$\tg_{\alpha_{s}}:\ux_{\alpha_{s}}\rightarrow\ux_{\alpha_{s+1}}$, which connect 
the simplicial complexes at consecutive scales.
We show that such maps are induced by $g_{\alpha_{s}}$. 

\begin{lemma}
\label{lemma:bary_activeimage}
Let $f$ be any active face of $\square_{\alpha_{s}}$. 
Then, $g_{\alpha_{s}}(f)$ is an active face of $\square_{\alpha_{s+1}}$.
\end{lemma}

\begin{proof}
Using Lemma~\ref{lemma:gcell}, $e:=g_{\alpha_{s}}(f)$ is a face 
of $\square_{\alpha_{s}}$. 
If $e$ is a vertex, then it is active, because $f$ contains
at least one active vertex $v$, and $g_{\alpha_{s}}(v)=e$ in this case.
If $e$ is not a vertex, we assume for a contradiction that it is not active.
Then, it contains a facet $e_1$ that contains all active vertices in $e$.
Let $e_2$ denote the opposite facet of $e_1$ in $e$. 
By Lemma~\ref{lemma:gcell}, $f$ contains
opposite facets $f_1$, $f_2$ such that $g_{\alpha_{s}}(f_1)=e_1$ 
and $g_{\alpha_{s}}(f_2)=e_2$.
Since $f$ is active, both $f_1$ and $f_2$ contain active vertices;
in particular, $f_2$ contains an active vertex $v$. 
But then the active vertex $g_{\alpha_{s}}(v)$ must lie in $e_2$, 
contradicting the fact that $e_1$ contains all active vertices of $e$.
\end{proof}

As a result, $g$ is well defined for each face $e\in\square_{\alpha_{s}}$,
since there exists some active face $e'\in\square_{\alpha_{s}}$ with
$e\subseteq e'$, and $g(e)\subseteq g(e')$.
By definition, a simplex $\sigma\in \ux_{\alpha_s}$ is a flag 
$(f_0\subseteq \ldots \subseteq f_k)$ of faces in $\square_{\alpha_{s}}$.
We set 
\[
\tg_{\alpha_{s}}(\sigma):=(g_{\alpha_{s}}(f_0),\ldots ,g_{\alpha_{s}}(f_k)),
\]
where $(g_{\alpha_{s}}(f_0)\subseteq \ldots \subseteq g_{\alpha_{s}}(f_k))$ 
is a flag 
of faces in $\square_{\alpha_{s+1}}$ by Lemma~\ref{lemma:bary_activeimage}, 
and hence is a simplex in $\ux_{\alpha_{s+1}}$.
It follows that $\tg_s:\ux_{\alpha_{s}}\rightarrow\ux_{\alpha_{s+1}}$ 
is a simplicial map.
This completes the description of the simplicial tower
\[
(\ux_{\alpha_{s}})_{s\in\Z}.
\]

\subsection{Interleaving with the Rips module}
\label{subsection:bary_interleaving}

First, we show that our tower is a constant-factor approximation 
of the the $\lin$-Rips filtration of $P$.
We then show the relation between our approximation tower and
the Euclidean Rips filtration of $P$.

We start by defining two acyclic carriers. 
First, we set $\lambda=1$ and 
abbreviate $\alpha:=\alpha_s=2^s$ to simplify notation.

\begin{itemize}
\item $C_1^\alpha:\rin_{\alpha/2} \rightarrow \ux_{\alpha}$: 
for any simplex $\sigma=(p_0,\ldots,p_k)$ in $\rin_{\alpha/2}$,
we set $C_1^\alpha(\sigma)$ as the barycentric span of 
$U:=\{a_s(p_0),\ldots,a_s(p_k)\}$, which is a subcomplex of $\ux_{\alpha}$.
Using Lemma~\ref{lemma:bary_iripscell}, 
$U$ lies in a maximal active face $f$ of 
$\square_\alpha$, so that $C_1^\alpha(\sigma)$ is acyclic.
The barycentric span of any subset of $U$ is a subcomplex of the barycentric
span of $U$, so $C_1^\alpha$ is a carrier.
Therefore, $C_1^\alpha$ is an acyclic carrier.

\item $C_2^\alpha:\ux_{\alpha}\rightarrow \rin_{\alpha}$: 
let $\sigma$ be any flag of $\ux_{\alpha}$ and let $E$ be the smallest
active face of $\square_\alpha$ that contains $\sigma$ (we break ties
by making use of an arbitrary global order $\succ$ on $P$)\footnote{
We define an order between the active faces of $\square_{\alpha}$, 
using $\succ$: for each active face $F\in \square_{\alpha}$, there are
at least two points of $P$ whose images under $g_{\alpha}$ are vertices of $F$; 
say $\{q_1\succ q_2\succ \ldots  \succ q_m\}\subseteq P$ are the points that
map to $F$.
We assign to $F$ the string of length $n$:
$q_1 q_2 \ldots q_m\overbrace{\text{\o},\ldots,\text{\o}}^{n-m}$.
Each active face has a unique string associated to it. 
A total order on the faces is obtained by taking the lexicographic orders of
the strings of each active face.
}.
We collect all the points of $P$ that map to vertices of $E$
under the map $a_{\alpha}$ and set $C_2^{\alpha}(\sigma)$ as the simplex 
on this set of points.
By an application of the triangle inequality, we see that the $L_\infty$-diam 
of $C_2^{\alpha}(\sigma)$ is at most $2\alpha$, 
so $C_2^{\alpha}(\sigma)\in  \rin_{\alpha}$ and is acyclic.
It is also clear that $C_2^{\alpha}(\tau)\subseteq C_2^{\alpha}(\sigma)$
for each $\tau \subseteq \sigma$, so $C_2^{\alpha}$ is an acyclic carrier.
\end{itemize}
Using the acyclic carrier theorem (Theorem~\ref{theorem:acyclic_carrier}),
there exist augmentation-preserving chain maps
\[
c_1^\alpha:\ch_\ast(\rin_{\alpha/2})\rightarrow \ch_\ast(\ux_{\alpha}) 
\quad \text{and}\quad 
c_2^\alpha:\ch_\ast(\ux_\alpha)\rightarrow \ch_\ast(\rin_{\alpha}),
\]
between the chain complexes, which are carried by $C_1^\alpha$ and 
$C_2^\alpha$ respectively, for each $\alpha\in I$. 
We obtain the following diagram of augmentation-preserving chain maps:

\begin{eqnarray}
\label{diagram:bary_inf_intlv}
\xymatrix{
 \ldots\ar[r] & \ch_\ast(\rin_{\alpha}) \ar[rd]^{c_1}\ar[r]^{inc} & 
\ch_\ast(\rin_{2\alpha}) \ar[r] &\ldots
\\
\ldots\ar[r] & \ch_\ast(\ux_{\alpha}) \ar[r]^{\tg}\ar[u]^{c_2} & 
\ch_\ast(\ux_{2\alpha })\ar[r] \ar[u]^{c_2} &\ldots
\\
}
\end{eqnarray}
where $inc$ corresponds to the chain map for inclusion maps,
and $\tg$ denotes the chain map for the corresponding
simplicial map $g$ (we removed indices of the maps for readability). 

The chain complexes give rise to a diagram of the corresponding homology groups,
connected by the induced linear maps $c_1^\ast,c_2^\ast,inc^\ast,\tg^\ast$:
\begin{eqnarray}
\label{equation:bary_hom}
\xymatrix{
 \ldots\ar[r] & H(\rin_{\alpha}) \ar[rd]^{c_1^\ast}\ar[r]^{inc^\ast} & 
H(\rin_{2\alpha}) \ar[r] &\ldots
\\
\ldots\ar[r] & H(\ux_{\alpha}) \ar[r]^{\tg^\ast}\ar[u]^{c_2^\ast} & 
H(\ux_{2\alpha })\ar[r] \ar[u]^{c_2^\ast} &\ldots
\\
}
\end{eqnarray}

\begin{lemma}
\label{lemma:bary_inf_intlv_low}
For all $\alpha\in I$, the linear maps in the lower triangle of
Diagram~\eqref{equation:bary_hom} commute, that is, 
\[
\tg^\ast=c_1^\ast\circ c_2^\ast.
\]
\end{lemma}

\begin{proof}
We look at the corresponding triangle in 
Diagram~\eqref{diagram:bary_inf_intlv}.
We show that the (augmentation-preserving) chain maps $\tg$ and $c_1\circ c_2$ 
are both carried by an acyclic carrier $D:\ux_\alpha\rightarrow \ux_{2\alpha}$.
The claim then follows from the acyclic carrier theorem.

Let $\sigma\in \ux_{\alpha}$ be any flag and let $E\in\square_{\alpha}$ 
denote the minimal active face containing $\sigma$.
Let $\{q_1,\dots,q_k \}$ be the active vertices of $E$.
Let $\{ p_1,\dots, p_m \}$ be the set of points of $P$ 
that map to $\{q_1,\dots,q_k \}$ under the map $a_{\alpha}$.
Since the $L_\infty$-diameter of $\{ p_1,\dots, p_m \}$ is at most $2\alpha$,
using Lemma~\ref{lemma:bary_iripscell} we see that 
$\{ a_{2\alpha}(p_1),\dots,a_{2\alpha}(p_m) \}$ is a face 
of $\square_{2\alpha}$.
We set $D(\sigma)$ as the barycentric span of 
$\{ a_{2\alpha}(p_1),\dots,a_{2\alpha}(p_m) \}$.
It follows that $D$ is an acyclic carrier.

Further, $\{ a_{2\alpha}(p_1),\dots,a_{2\alpha}(p_m) \}=
\{ g_{2\alpha}(q_1),\dots,g_{2\alpha}(q_k) \}$
from Lemma~\ref{lemma:vorcontain}, 
so $D(\sigma)$ is the barycentric subdivision of $g_{2\alpha}(E)$.
As a result $D=C_1\circ C_2$
so that it carries $c_1\circ c_2$.
We show that $D$ also carries the map $\tg$.

By definition, for each face $e\subseteq E$, $g(e)\subseteq g(E)$
and $\tg(sd(e))\subseteq \tg(sd(E))$.
This means that $\tg(\sigma)$ is contained in $g(E)$.
This shows that $\tg(\sigma)\in C_1\circ C_2 (\sigma)$
implying that $\tg$ is carried by $C_1\circ C_2$, as required.
\end{proof}

\begin{lemma}
\label{lemma:bary_inf_intlv_up}
For all $\alpha\in I$, the linear maps in the upper triangle 
of Diagram~\eqref{equation:bary_hom} commute, that is, 
\[
inc^\ast=c_2^\ast\circ c_1^\ast.
\]
\end{lemma}

\begin{proof}
The proof technique is analogous to the proof 
of Lemma~\ref{lemma:bary_inf_intlv_low}.
We define an acyclic carrier $D:\rin_{\alpha}\rightarrow\rin_{2\alpha}$ 
which carries $inc$ and $c_2\circ c_1$, both of which are
augmentation-preserving.

Let $\sigma=(p_0,\ldots,p_k)\in \rin_{\alpha}$ be any simplex. 
The set of active vertices 
\[
U:=\{a_{2\alpha}(p_0),\ldots,a_{2\alpha}(p_k)\}\subset G_{2\alpha}
\] 
lie in a face $f$ of $G_{2\alpha}$, using Lemma~\ref{lemma:bary_iripscell}.
We can assume that $f$ is active, as otherwise, we argue about a facet
of $f$ that contains $U$.
We set $D(\sigma)$ as the simplex on the subset of points in $P$,
whose closest grid point in $G_{2\alpha}$ is any vertex of $f$.
Using the triangle inequality we see that
$D(\sigma)\in \rin_{2\alpha}$, so $D$ is an acyclic carrier.
The vertices of $\sigma$ are a subset of $D(\sigma)$, 
so $D$ carries the map $inc$. 
Showing that $D$ carries $c_2\circ c_1$ requires further explanation. 

Let $\delta$ be any simplex in $\ux_{2\alpha}$ for which the chain
$c_1(\sigma)$ takes a non-zero value. 
Since $c_1(\sigma)$ is carried by $C_1(\sigma)$, we have that
$\delta\in C_1(\sigma)$, which is the barycentric span of $U$. 
Furthermore, for any $\tau\in C_1(\sigma)$, $C_2(\tau)$ is a simplex on 
the set of vertices $\{p \in P \mid a_{2\alpha}(p)\in V(f)\}$.
It follows that $C_2(\tau)\subseteq D(\sigma)$. 
In particular, since $c_2$ is carried by $C_2$, 
$c_2(c_1(\sigma))\subseteq D(\sigma)$ as well.
\end{proof}

Using Lemma~\ref{lemma:bary_inf_intlv_low} and 
Lemma~\ref{lemma:bary_inf_intlv_up},
we see that the two persistence modules 
$\left(H(\ux_{{\alpha_{s}}})\right)_{s\in\Z}$ 
and $\left(H(\rin_\alpha)\right)_{\alpha\ge 0}$ are weakly $2$-interleaved.

With elementary modifications in the definition of $\ux$ and $\tg$, 
we can get a tower of the form $(\ux_\alpha)_{\alpha\ge 0}$. 
Furthermore, with minor changes in the interleaving arguments, we show that the 
corresponding persistence module is strongly 4-interleaved with the 
$\lin$-Rips module.
Using scale balancing, this result improves to a strong 2-interleaving
(see Lemma~\ref{lemma:bary_strong_full}).
Since the techniques used in the proof are very similar to the concepts used 
in this section, for the sake of brevity we defer all further details to 
Appendix~\ref{subsection:appendix-strong}. 

Using the strong stability theorem for persistence modules 
and taking scale balancing into account, we immediately get that:
\begin{theorem}
\label{theorem:bary_irips_ratio}
The scaled persistence module $\big(H(\ux_{2\alpha})\big)_{\alpha\ge 0}$ 
and the $L_\infty$-Rips persistence 
module $\big(H(\rin_\alpha)\big)_{\alpha\ge 0}$
are $2$-approximations of each other.
\end{theorem}

For any pair of points $p,p'\in \R^d$, it holds that 
\[
\|p-p'\|_\infty\le \|p-p'\|_2 \le \sqrt{d}\,\|p-p'\|_\infty.
\]
This in turn shows that the $L_2$- and the 
$L_\infty$-Rips filtrations are strongly $\sqrt{d}$-interleaved.
Using the scale balancing technique for strongly interleaved
persistence modules, we get:
\begin{lemma}
\label{lemma:bary_rips_relation}
The scaled persistence module $(H(\ri_{\alpha/d^{0.25}}))_{\alpha\ge 0}$ 
and $(H(\rin_\alpha))_{\alpha\ge 0}$ are strongly $d^{0.25}$-interleaved.
\end{lemma}

Using Theorem~\ref{theorem:bary_irips_ratio}, Lemma~\ref{lemma:bary_rips_relation}
and the fact that interleavings satisfy the triangle 
inequality~\cite[Theorem\,~3.3]{bs-categorization}, we see that 
the module $(H(\ux_{2\alpha}))_{\alpha\ge 0}$ is strongly $2d^{0.25}$-interleaved with 
the scaled Rips persistence module $(H(\ri_{\alpha/d^{0.25}}))_{\alpha\ge 0}$.
We can remove the scaling in the Rips filtration simply by multiplying 
the scales on both sides with $d^{0.25}$ and obtain our 
final approximation result:

\begin{theorem}
\label{theorem:bary_rips_ratio}
The persistence module 
$\big(H(\ux_{2\sqrt[4]{d}\alpha})\big)_{\alpha\ge 0}$ and the Euclidean Rips 
persistence module $\big(H(\ri_{\alpha})\big)_{\alpha\ge 0}$
are $2d^{0.25}$-approximations of each other.
\end{theorem}

\section{Computational complexity}
\label{section:bary_compute}

In this section, we discuss the computational aspects of constructing the 
approximation tower.
In Subsection~\ref{subsection:bary_size} we discuss the size 
complexity of the tower.
An algorithm to compute the tower efficiently is presented in 
Subsection~\ref{subsection:bary_computing}.

\paragraph{Range of relevant scales.}

Set $n:=|P|$ and let $CP(P)$ denote the closest pair distance of $P$. 
At scale $\alpha_0:=\frac{CP(P)}{3d}$ and lower, no two active vertices
lie in the same face of the grid,
so the approximation complex consists of $n$ isolated $0$-simplices.
At scale $\alpha_m:=diam(P)$ and higher, points of $P$ map to active vertices
of a common face (by Lemma~\ref{lemma:bary_iripscell}), so 
the generated complex is acyclic. 
We inspect the range of scales $[\alpha_0,\alpha_m]$ to construct the tower, 
since the barcode is explicitly known for scales outside this range.
For this, we set $\lambda=\alpha_0$ in the definition of the scales.
The total number of scales is 
\[
\ceil{\log_2 \alpha_m/\alpha_0}=\left\lceil \log_2 \frac{diam(P)3d}{CP(P)} \right\rceil=\ceil{\log_2\Delta +\log_2 3d}=O(\log\Delta+\log d),
\]
where $\Delta=\frac{diam(P)}{CP(P)}$ is the spread of the point set.

\subsection{Size of the tower}
\label{subsection:bary_size}

The size of a tower is the number of simplices that do not have a preimage,
that is, the number of simplex inclusions in the tower.
We start by counting the number of active faces used in the tower.
\begin{lemma}
\label{lemma:num_activefaces}
The number of active faces without pre-image in the tower is at most $n3^d$.
\end{lemma}

\begin{proof}
At scale $\alpha_0$, there are $n$ inclusions of $0$-simplices in the tower, 
due to $n$ active vertices. 
Using Lemma~\ref{lemma:vorcontain}, $g$ is surjective
on the active vertices of $\square$ (for any scale). 
Hence, no further active vertices are added to the tower.

It remains to count the maximal active faces of dimension $\geq 1$ 
without preimage.
We will use a charging argument, charging the existence of such an active face
to one of the points in $P$.
We show that each point of $P$ is charged at most $3^{d}-1$ times, 
which proves the claim.
For that, we first fix an arbitrary total order $\prec$ on $P$. 
Each active vertex on any scale has a non-empty subset of $P$ in its Voronoi 
region; we call the maximal such point with respect to the order $\prec$
the \emph{representative} of the active vertex. 

For each active face $f$ of dimension at least one, 
we define the \emph{signature} of $f$ as the set of representatives 
of the active vertices of $f$.
If for any set of active vertices $u_1,\dots,u_k$ we have that 
$v=g(u_1)=\dots=g(u_k)$, then the representative of $v$ is 
one of the representatives of $u_1,\dots,u_k$, 
using Lemma~\ref{lemma:vorcontain}.
Therefore, the signatures of the active faces that are images of $f$ 
under $g$ are subsets of the signature of $f$.
This implies that each maximal active face that is included has a unique
maximal signature.
We bound the number of maximal signatures to get
a bound on the number of maximal active face inclusions.
We charge the addition of each maximal signature to 
the lowest ordered point according to $\prec$.

Each signature contains representatives of active vertices
from a face of $\square_\alpha$.
Since each active vertex $v$ has $3^d-1$ neighboring vertices in the grid
that lie in a common face, the representative $p$ of $v$
can be charged $3^d-1$ times.
There is a canonical isomorphism between the neighboring vertices
of $v$ at each scale.
Then, for $p$ to be charged more times, 
the image of $v$ and some neighboring vertex
$u$ must be identical under $g$ at some scale.
But then, the representative of $g(v)=g(u)$ is not $p$ anymore, 
since $p$ was the lowest ranked point in its neighborhood,
hence the representative changes when the Voronoi regions are combined.
So, $p$ could not have been charged in such a case.
Therefore, each point $p\in P$ is indeed charged at most $3^d-1$ times.

There are $n$ active faces of dimension $0$ and at most $n(3^d-1)$
active faces of higher dimension.
The upper bound is $n+n(3^d-1)=n3^{d}$, as claimed.
\end{proof}

\begin{theorem}
\label{theorem:towersize}
The $k$-skeleton of the tower has size at most 
\[
n6^{d-1}(2k+4)(k+3)! \left\{\begin{array}{c}d\\k+2\end{array}\right\}
=n2^{O(d\log k + d)},
\]
where
$ \left\{\begin{array}{c}a\\b\end{array}\right\}$ denotes the Stirling
number of the second kind.
\end{theorem}

\begin{proof}  
Each $k$-simplex that is included in the tower at any given scale $\alpha$
is a part of the barycentric subdivision of an active
face that is also included at $\alpha$.
Therefore, we can account for the inclusion of this simplex
by including the barycentric subdivision of its parent active face.

From Lemma~\ref{lemma:num_activefaces} at most $n3^d$
active faces are included in the tower over all dimensions.
We bound the number of $k$-simplices in the barycentric
subdivision of a $d$-cube.
Multiplying with $n3^{d}$ gives the required bound.

Let $c$ be any $d$-cube of $\square_\alpha$.
To count the number of flags of length $(m+1)$ contained in $c$
that start with some vertex and end with $c$, we use similar ideas 
as in~\cite{ek-diagonal}: first, we fix any vertex
$v$ of $c$ and count the flags of the form $v\subseteq\ldots\subseteq c$. 
Every $\ell$-face in $c$ incident to $v$ corresponds to a subset of $\ell$ 
coordinate indices, in the sense that the coordinates not chosen are fixed 
to the coordinates of $v$ for the face. 
With this correspondence, a flag from $v$ to $c$ of 
length $(m+1)$ corresponds to an ordered $m$-partition of $\{1,\ldots,d\}$.
The number of such partitions is known as $m!$ times the quantity
$\left\{\begin{array}{c}d\\m\end{array}\right\}$, 
which is the Stirling number of second kind~\cite{rd-stirling}, 
and is upper bounded by $2^{O(d\log m)}$. 
Since $c$ has $2^d$ vertices, the total number
of flags $v\subseteq\ldots\subseteq c$ of length $(m+1)$ with any vertex $v$ 
is hence $2^d m! \left\{\begin{array}{c}d\\m\end{array}\right\}$.

We now count the number of flags of length $k+1$.
Each such flag is $(k+1)$-subset of some flag of length $m=k+3$ that start with a vertex
and end with $c$.
There are $2^d (k+2)! \left\{\begin{array}{c}d\\k+2\end{array}\right\}$ such flags
and each of them has $\binom{k+3}{k+1}=(k+3)(k+2)/2$ subsets of size  $(k+1)$.
The number of $(k+1)$-flags is upper bounded by
$2^d (k+2)! \left\{\begin{array}{c}d\\k+2\end{array}\right\}\frac{(k+3)(k+2)}{2}
=2^{d-1} (k+2)(k+3)! \left\{\begin{array}{c}d\\k+2\end{array}\right\}$.
The $k$-skeleton has size at most
\[
n3^{d}2^{d-1} (k+2)(k+3)! \left\{\begin{array}{c}d\\k+2\end{array}\right\}
=n6^{d-1}(2k+4)(k+3)! \left\{\begin{array}{c}d\\k+2\end{array}\right\}.
\]
\end{proof}

\subsection{Computing the tower}
\label{subsection:bary_computing}

From Section~\ref{section:shifted_grids},
we know that $G_{\alpha_{s+1}}$ is built from $G_{\alpha_{s}}$ by making use of 
an arbitrary translation vector $(\pm 1,\ldots,\pm 1)\in\Z^d$.
In our algorithm, we pick the components of this translation vector 
uniformly at random from $\{+1,-1\}$, and independently for each scale.
The choice behind choosing this vector randomly becomes more
clear in the next lemma.

From the definition, the cubical maps 
$g_{\alpha_{s}}:\square_{\alpha_{s}}\rightarrow\square_{\alpha_{s+1}}$ 
can be composed for multiple scales.
For a fixed ${\alpha_{s}}$, we denote by 
$g^{(j)}:\square_{\alpha_{s}}\rightarrow\square_{\alpha_{s+j}}$ 
the $j$-fold composition of $g$, that is,
\[
g^{(j)}=g_{\alpha_{s+j-1}}\circ g_{\alpha_{s+j-2}}\circ\ldots \circ 
g_{\alpha_{s+1}}\circ g_{\alpha_{s}},
\]
for $j\ge 1$.

\begin{lemma}
\label{lemma:bary_g_survival}
For any $k$-face $f\in \square_{\alpha_{s}}$ with $1\le k\le d$, let $Y$ denote 
the minimal integer $j$ such that $g^{(j)}(f)$ is a vertex, for a 
given choice of the randomly chosen translation vectors.
Then, the expected value of $Y$ satisfies 
\[
\E[Y]\leq 3\log k,
\]
which implies that no face of $\square_{\alpha_{s}}$ \emph{survives} 
more than $3\log d$ scales in expectation.
\end{lemma}

\begin{proof}
Without loss of generality, assume that the grid under consideration is $\Z^d$
and $f$ is the $k$-face spanned by the vertices 
$\{\underbrace{\{0,1\},\ldots,\{0,1\}}_{k},0,\ldots,0\}$,
so that the origin is a vertex of $f$.
The proof for the general case is analogous.

Let $y_1\in\{-1,1\}$ denote the randomly chosen first coordinate of the 
translation vector, so that the corresponding shift is one of $\{ -1/2,1/2 \}$. 
\begin{itemize}
\item If $y_1=1$, then the grid $G'$ on the next scale has some grid point 
with $x_1$-coordinate $1/2$. 
Clearly, the closest grid point in $G'$ to the origin is of the form 
$(+1/2,\pm 1/2,\ldots,\pm 1/2)$, and thus, 
this point is also closest to $(1,0,0,\ldots,0)$. 
The same is true for any point $(0,\ast,\ldots,\ast)$ and its corresponding point 
$(1,\ast,\ldots,\ast)$ on the opposite facet of $f$.
Hence, for $y_1=1$, $g(f)$ is a face where all points 
have the same $x_1$-coordinate.

\item On the other hand, if $y_1=-1$, the origin is mapped to some point 
which has the form
$(-1/2,\pm 1/2,\ldots,\pm 1/2)$ and $(1,0,\ldots,0)$ is mapped to 
$(3/2,\pm 1/2,\ldots,\pm 1/2)$, as one can directly verify.
Hence, in this case, in $g(f)$, points do not all have the same $x_1$ coordinate.
\end{itemize}
We say that the $x_1$-coordinate \emph{collapses} in the first case and 
\emph{survives} in the second.
Both events occur with the same probability $1/2$.
Because the shift is chosen uniformly at random for each scale, 
the probability that $x_1$ did not collapse after $j$ iterations
is $1/2^{j}$.

$f$ spans $k$ coordinate directions, so it must
collapse along each such direction to contract to a vertex.
Once a coordinate collapses, it stays collapsed at all higher scales.
As the random shift is independent for each coordinate
direction, the probability of a collapse is the same 
along all coordinate directions that $f$ spans.
Using the union bound, the probability that $g^j(f)$ has not collapsed 
to a vertex is at most $k/2^j$. 
With $Y$ as in the statement of the lemma, it follows that
\[
P(Y\ge j)\le k/2^j.
\] 
Hence, 
\begin{align*}
	\E[Y]=&\sum_{j=1}^{\infty} j P(Y=j) = \sum_{j=1}^{\infty} P(Y\ge j) \\
	&\le  \log k + \sum_{c=1}^{\infty}\sum_{j=c\log k}^{(c+1)\log k} P(Y\ge j)\\
	& \le \log k + \sum_{c=1}^{\infty}\sum_{j=c\log k}^{(c+1)\log k}
	P(Y\ge c\log k)\\
	&\le \log k+ \sum_{c=1}^{\infty}\log k\frac{k}{2^{c \log k}} \\
	& \le \log k+ \log k\sum_{c=1}^{\infty} \frac{1}{k^{c-1}} \\
	&\le \log k + 2\log k \le 3 \log k. 
\end{align*}

\end{proof}

As a consequence of the lemma, the expected ``lifetime'' of $k$-simplices 
in our tower with $k>0$ is rather short: given a flag 
$e_0\subseteq\ldots\subseteq e_\ell$, the face $e_\ell$ will be mapped to a vertex
after $O(\log d)$ steps, and so will be all its sub-faces, 
turning the flag into a vertex.
It follows that summing up the total number of $k$-simplices with $k>0$ over  
$\ux_\alpha$ for all $\alpha\ge 0$ yields an upper bound of $n2^{O(d\log k +d)}$ as well.

\subsubsection*{Algorithm description}
\label{para:bary_algo1}

Recall that a simplicial map can be written 
as a composition of simplex inclusions and contractions of 
vertices~\cite{dfw-gic,ks-twr}. 
That means, given the complex $\ux_{\alpha_s}$, to describe the 
complex at the next scale $\alpha_{s+1}$, it suffices to specify
\begin{itemize}
\item which pairs of vertices in $\ux_{\alpha_s}$ 
map to the same image under $\tg$, and

\item which simplices in $\ux_{\alpha_{s+1}}$ are included 
at scale $\ux_{\alpha_{s+1}}$. 
\end{itemize}

The input is a set of $n$ points $P\subset \R^d$.
The output is a list of \emph{events}, where each event is of 
one of the three following types:
\begin{itemize}
\item A \emph{scale event} defines a real value $\alpha$ and 
signals that all upcoming events 
happen at scale $\alpha$ (until the next scale event).

\item An \emph{inclusion event} introduces a new simplex, 
specified by the list of vertices on its boundary (we assume that 
every vertex is identified by a unique integer).

\item A \emph{contraction event} is a pair of vertices $(i,j)$ from the previous
scale, and signifies that $i$ and $j$ are identified as the same from that scale.
\end{itemize}
 
In a first step, we estimate the range of scales that we are interested in.
We compute a $2$-approximation of $diam(P)$ 
by taking any point $p\in P$ and calculating $\max_{q\in P}\|p-q\|$.
Then we compute $CP(P)$ using 
a randomized algorithm in $n2^{O(d)}$ expected time~\cite{km-closest}.

Next, we proceed scale-by-scale and construct the list of events accordingly.
On the lowest scale, we simply compute the active vertices 
by point location for $P$
in a cubical grid, and enlist $n$ inclusion events
(this is the only step where the input points are considered in the algorithm).

For the data structure, we use an auxiliary container $S$ 
and maintain the invariant that whenever a new scale is considered, 
$S$ consists of all simplices of the previous scale,
sorted by dimension. 
In $S$, for each vertex, we store an id and
a coordinate representation of the active face to which it corresponds.
Every $\ell$-simplex with $\ell>0$ is stored just as a list of integers, 
denoting its boundary vertices.
We initialize $S$ with the $n$ active vertices at the lowest scale.

Let $\alpha<\alpha'$ be any two consecutive 
scales with $\square,\square'$ the respective cubical complexes
and $\ux,\ux'$ the approximation complexes, with $\tg:\ux\rightarrow \ux'$ 
being the simplicial map connecting them. 
Suppose we have already constructed all events at scale $\alpha$. 
\begin{itemize}
\item First, we enlist the scale event for $\alpha'$.

\item Then, we enlist the contraction events. 
For that, we iterate through the vertices of $\ux$ and 
compute their value under $g$, using point location in a cubical grid.
We store the results in a list $S'$ (which contains the simplices of $\ux'$).
If for a vertex $j$, $g(j)$ is found to be equal to $g(i)$ for a 
previously considered vertex $i$, we choose the minimal such $i$ and 
enlist a contraction event for $(i,j)$.

\item We turn to the inclusion events:
\begin{itemize}
\item We start with the case of vertices.
Every vertex of $\ux'$ is either an active face or a 
secondary face of $\square'$.
Each active face must contain an active vertex, which is also a vertex 
of $\ux'$. 
We iterate through the elements in $S'$. 
For each active vertex $v$ encountered, we go over all faces of the 
cubical complex $\square'$ that contain $v$ as a vertex, 
and check whether they are active. 
For every active face $E$ encountered that is not in $S'$ yet, we add it to $S'$
and enlist an inclusion event of a new $0$-simplex. 
Additionally, we go over each face of $E$, add it to $S'$ and 
enlist a vertex inclusion event,
thereby enumerating the secondary faces that are in $E$.
At termination, all vertices of $\ux'$ have been detected.

\item Next, we iterate over the simplices of $S$ of dimension $\geq 1$, 
and compute their image under $\tg$	using the pre-computed vertex map;
we store the result in $S'$.

\item To find the simplices of dimension $\geq 1$ included at $\ux'$, 
we exploit our previous insight that they 
contain at least one vertex that is included 
at the same scale (see the proof of Theorem~\ref{theorem:towersize}).
Hence, we iterate over the vertices included in $\ux'$ and find the 
included simplices inductively in dimension.

Let $v$ be the current vertex under consideration;
assume that we have found all $(p-1)$-simplices in $\ux'$ that contain $v$. 
Each such $(p-1)$-simplex $\sigma$ is a flag of length $p$ in $\square'$. 
We iterate over all faces $e$ that extend $\sigma$ to a flag of length $p+1$.
If $e$ is active, we have found a $p$-simplex in $\ux'$ incident to $v$. 
If this simplex is not in $S'$ yet, 
we add it and enlist an inclusion event for it. 
We also enqueue the simplex in our inductive procedure, to look for
$(p+1)$-simplices in the next round. 
At the end of the procedure, we have detected all simplices in $\ux'$ 
without preimage, and $S'$ contains all simplices of $\ux'$. 
We set $S\gets S'$ and proceed to the next scale. 
\end{itemize}
\end{itemize}
This ends the description of the algorithm.
	
\begin{theorem}
\label{theorem:algo1_bary}
To compute the $k$-skeleton, the algorithm takes
\[
n2^{O(d)}\log\Delta + 2^{O(d)}M
\] 
time in expectation and $M$ space, where $M$ denotes the size of the tower. 
In particular, the expected time is bounded by 
\[
n2^{O(d)}\log\Delta + n2^{O(d\log k +d)}
\]
and the space is bounded by $n2^{O(d\log k +d)}$.
\end{theorem}

\begin{proof}
In the analysis, we ignore the costs of point locations in grids,
checking whether a face is active, and searches in data structures $S$,
since all these steps have negligible costs when
appropriate data structures are chosen.

Computing the image of a vertex of $\ux$ costs $O(2^d)$ time.
Moreover, there are at most $n2^{O(d)}$ vertices altogether in the tower
in expectation (using Lemma~\ref{lemma:num_activefaces}), 
so this bound in particular holds on each scale.
Hence, the contraction events on a fixed scale
can be computed in $n2^{O(d)}$ time.
Finding new active vertices requires iterating over the cofaces of a vertex 
in a cubical complex. 
There are $3^d$ such cofaces for each vertex. 
This has to be done for a subset of the vertices
in $\ux'$, so the running time is also $n2^{O(d)}$.
Further, for each new active face, we go over its $2^{O(d)}$ faces to 
enlist the secondary faces, so this step also consumes $n2^{O(d)}$ time.
Since there are $O(\log\Delta+\log d)$ scales considered, 
these steps require $n2^{O(d)}\log\Delta$ over all scales.

Computing the image of $\tg$ for a fixed scale costs at most $O(2^d|\ux|)$.
$M$ is the size of the tower, that is, the simplices without preimage, 
and $I$ is the set of scales considered. 
The expected bound for $\sum_{\alpha\in I} |\ux_\alpha|=O(\log d M)$, 
because every simplex has an expected lifetime of at most $3\log d$ by 
Lemma~\ref{lemma:bary_g_survival}. 
Hence, the cost of these steps is bounded by $2^{O(d)}M$.

In the last step of the algorithm, we find the simplices of $\ux'$
included at $\alpha'$. 
We consider a subset of simplices of $\ux'$, and for each, we iterate over
a collection of faces in the cubical complex of size at most $2^{O(d)}$. 
Hence, this step is also bounded by $2^{O(d)}|\ux|$ per scale, 
and hence bounded $2^{O(d)}M$ as well.

For the space complexity, the auxiliary data structure $S$ 
gets as large as $\ux$, which is clearly bounded by $M$.
For the output complexity, the number of contraction events is at most the 
number of inclusion events, because every contraction removes a vertex 
that has been included before.
The number of inclusion events is the size of the tower. 
The number of scale events as described is $O(\log\Delta+\log d)$. 
However, it is simple to get rid of this factor by only including scale events 
in the case that at least one inclusion or contraction
takes place at that scale. 
The space complexity bound follows.
\end{proof}

\subsection{Dimension reduction}
\label{subsection:bary_dimred}

When the ambient dimension $d$ is large, our approximation scheme can be 
combined with dimension reduction techniques to reduce the final complexity, 
very similar to the application in~\cite{ckr-polynomial-dcg}. 
For a set of $n$ points $P\subset \R^d$, we apply the dimension reduction
schemes of Johnson-Lindenstrauss~(JL)~\cite{jl-lemma},
Matou\v{s}ek~(MT)~\cite{mt-embedding}, 
and Bourgain's embedding~(BG)~\cite{bg-metric}.
We then compute the approximation on the lower-dimensional point set.
We only state the main results in Table~\ref{table:bary_dimred}, 
leaving out the proofs since they are very similar to those from~\cite{ckr-polynomial-dcg}.

\begin{table}[ht]
\centering
\begin{tabular}{| c | c | c | c |}
\hline
technique & approximation ratio & size & runtime	
\\ \hline
JL & $O(\log^{0.25} n)$ & 
$n^{O(\log k)}$ & $n^{O(1)}\log \Delta+n^{O(\log k)}$ 
\\ \hline
MT & $O((\log n)^{0.75} (\log\log n)^{0.25})$ &
$n^{O(1)}$ & $n^{O(1)}\log\Delta$ 
\\ \hline
BG + MT & $O((\log n)^{1.75} (\log\log n)^{0.25})$ & 
$n^{O(1)}$ & $n^{O(1)}\log\Delta$ 
\\ \hline
\end{tabular}
\caption{Comparison of dimension reduction techniques:
here the approximation ratio is for the Rips persistence module,
and the size refers to the size of the $k$-skeleton of the approximation.}
\label{table:bary_dimred}
\end{table}

\section{Approximation scheme with Cubical complexes}
\label{section:cubical_scheme}

We extend our approximation scheme to use cubical complexes
in place of simplicial complexes.
We start by detailing a few aspects of cubical complexes.

\subsection{Cubical Complexes}
\label{subsection:cubical_cpx}

We now briefly describe the concept of cubical complexes, essentially
expanding upon the contents of Subsection~\ref{subsection:cubical-basic}.
For a detailed overview of cubical homology, we refer to~\cite{kmm-cubical}.

\subsubsection{Definition}
We define cubical complexes over the grids $G_{\alpha_{s}}$.
For any fixed $\alpha_{s}$, the grids $G_{\alpha_{s}}$ defines 
a natural collection of cubes.
An \emph{elementary cube} $\gamma$ is a product of intervals
$\gamma=I_1\times I_2\times \ldots \times I_d$, where each interval is of the form
$I_j=(x_j,x_j+m_j)$, such that the vertex $(x_1,\ldots,x_m)\in G_{\alpha_{s}}$ 
and each $m_j$ is either $0$ or $\alpha_s$.
That means, an (elementary) cube is simply a face of a $d$-cube of the grid.
An interval $I_j$ is said to be \emph{degenerate} if $m_j=0$.
The dimension of $\gamma$ is the number of non-degenerate intervals that
defines it.
We define the boundary of any interval as the two degenerate intervals
that form its endpoints and 
denote this by $\partial(I_j)=(x_j,x_j) + (x_j+m_j,x_j+m_j)$.
Taking the boundary of any fixed subset of the intervals defining $\gamma$ 
consecutively gives a sum of faces of $\gamma$.
A \emph{cubical complex} of $G_{\alpha_{s}}$ is a finite collection of cubes 
of $G_{\alpha_{s}}$.

We define chain complexes for the cubical case in the same way
as in simplicial complexes.
The chain complexes are connected by boundary homomorphisms,
where the boundary of a cube is defined as:
\[
\partial(I_1\times\ldots\times I_d) = 
(\partial(I_1)\times I_2 \times \ldots \times I_d) + \ldots + 
(I_1\times \ldots \times I_{d-1}\times \partial(I_d)),
\]
where $(I_1\times \ldots \times \partial(I_j)\times \ldots \times I_d)$ 
denotes the sum
\[
\Big(I_1\times \ldots \times (x_i,x_i) \times \ldots \times I_d\Big)+ 
\Big(I_1\times \ldots \times (x_i+m_i,x_i+m_i) \times \ldots \times I_d\Big).
\]
It can be quickly verified that for each cube $\gamma$, 
$\partial\circ \partial(\gamma) =0$ since each term appears twice 
in the expression and the addition is over $\Z_2$.

\subsubsection{Cubical maps and induced homology}

Let $T_{\alpha_{s}}$ and $T_{\alpha_{t}}$ denote the cubical complexes 
defined by the grids $G_{\alpha_{s}}$ and $G_{\alpha_{t}}$, respectively, 
for $s\le t$.
We use the vertex map $g:G_{\alpha_{s}}\rightarrow G_{\alpha_{t}}$ to define 
a map between the cubical complexes.
Note that if $(a,b)$ are vertices of a cube of $T_{\alpha_{s}}$ that differ 
in one coordinate, then $(g(a),g(b))$ are vertices of a cube of 
$T_{\alpha_{t}}$ that differ in at most one coordinate.
A \emph{cubical map} is a map $f:T_{\alpha_{s}}\rightarrow T_{\alpha_{t}}$ 
defined using $g$, such that for each cube 
$\gamma=[a_1,b_1]\times\ldots\times[a_d,b_d]$ of $T_{\alpha_{s}}$,
$f(\gamma):=[g(a_1),g(b_1)]\times\ldots\times[g(a_d),g(b_d)]$ spans a cube
of $T_{\alpha_{t}}$.
The cubical map can also be restricted to sub-complexes of $T_{\alpha_{s}}$ 
and $T_{\alpha_{t}}$, provided that the image $f(\gamma)$ is well-defined.

Each cubical map also defines a corresponding continuous map 
between the underlying spaces of the respective complexes.
Let $x\in |\gamma|$ be a point in $\gamma$.
Then, the coordinates of $x$ can be uniquely written as 
$x=[\lambda_1 a_1 + (1-\lambda_1) b_1,\ldots, \lambda_d a_d + (1-\lambda_d) b_d]$
where each $\lambda_i\in [0,1]$.
The image of $x$ under the continuous extension of $f$ is the point
$[\lambda_1 g(a_1) + (1-\lambda_1) g(b_1),\ldots, \lambda_d g(a_d) + (1-\lambda_d) g(b_d)]$ in the cube $g(\gamma)$.

The cubical map $f$ gives rise to a chain map $f_\#:C_p(T_{\alpha_{s}})
\rightarrow C_p(T_{\alpha_{t}})$
between the $p$-th chain groups of the complexes, for each $p\in[0,\ldots,d]$. 
For each cube $\gamma$, $f_\#(\gamma)=f(\gamma)$ if $dim(\gamma)=dim(f(\gamma))$
and $0$ otherwise.
For any chain $c=\sum_i \gamma_i$, the chain map is defined linearly 
$f_\#(c)=\sum_i f_\#(\gamma_i)$.
It is simple to verify that $\partial\circ f_\# = f_\#\circ \partial$, so 
this gives a homomorphism between the chain groups.

Moving to the homology level, we get the respective homology groups 
$H(T_{\alpha_{s}})$ and $H(T_{\alpha_{t}})$ and the chain map from above 
induces a linear map between them.
The concept of reduced homology and augmentation maps is also
applicable to the cubical chain complexes.
For a sequence of cubical complexes connected with cubical maps, this 
generates a persistence module.

Cubical filtrations and towers are defined in a similar manner to the 
simplicial case.
A \emph{cubical filtration} is a collection of cubical complexes
$(T_\alpha)_{\alpha\in I}$ such that $T_\alpha\subseteq T_\alpha'$ 
for all $\alpha\leq\alpha'\in I$. 
A (cubical) tower is a sequence $(T_\alpha)_{\alpha\in J}$ of cubical 
complexes with $J$ being an index set together with cubical maps 
between complexes at consecutive scales.
A cubical tower can be written as a sequence of inclusions and contractions,
where an inclusion refers to the addition of a cube and a contraction
refers to collapsing a cube along a coordinate direction to either of the
endpoints of the interval.

\subsection{Description}
\label{subsection:cubical_description}

We choose the simplest possible cubical complex to define our approximation
cubical tower: for each scale $\alpha_{s}$, we define the cubical complex
$U_{\alpha_{s}}$ as the 
set of active faces and secondary faces spanned by $V_{\alpha_{s}}$.
Hence the cubical complex is closed under taking faces and is well-defined.
See Figure~\ref{figure:bary_cpx_2} for a simple example.

Recall from Section~\ref{section:simplicial_scheme} that for each $s\in\Z$, 
$U_{\alpha_{s}}$ and $U_{\alpha_{s+1}}$ are related by a cubical map 
$g_{\alpha_{s}}$, which gives rise to the cubical tower
\[
(U_{\alpha_{s}})_{s\in\Z}.
\]
We extend this to a tower $(U_{\alpha})_{\alpha\ge 0}$ by using techniques
from Appendix~\ref{subsection:appendix-strong}.
In Section~\ref{section:simplicial_scheme} we saw that the tower
$(\ux_{\alpha})_{\alpha\ge 0}$ gives an approximation to the Rips filtration.
The relation between the simplicial and cubical towers is trivial:
$\ux_{\alpha_{s}}$ is simply a triangulation of $|U_{\alpha_{s}}|$. 
Hence $\ux_{\alpha_{s}}$ and $U_{\alpha_{s}}$ have the same 
homology~\cite{munkres}.
Moreover, the simplicial map is derived from an application of the cubical map.
In particular, the continuous versions of both maps are the same.
For any $0\le \alpha \le \beta $, let
\begin{itemize}
\item $f_1:H_\ast(U_{\alpha})\rightarrow H_\ast(U_{\beta})$ denote the homomorphism
induced by the cubical map, 

\item $f_2:H_\ast(\ux_{\alpha})\rightarrow H_\ast(\ux_{\beta})$ denote the homomorphism
induced by the simplicial map, and

\item $f_0:H_\ast(|\ux_{\alpha}|=|U_{\alpha}|)\rightarrow 
H_\ast(|\ux_{\beta}|=|U_{\beta}|)$ denote 
the homomorphism induced by the common continuous map.
\end{itemize}
It is well-established that $f_1=f_0$~\cite[Chapter.~6]{kmm-cubical}
and $f_2=f_0$~\cite[Chapter.~2]{munkres}.
Therefore, we conclude that the persistence modules 
$\big(H(U_{\alpha})\big)_{\alpha\ge 0}$ 
and $\big(H(\ux_{\alpha})\big)_{\alpha\ge 0}$ are persistence-equivalent.
Combining this observation with the result of 
Theorem~\ref{theorem:bary_rips_ratio}, we get

\begin{theorem}
\label{theorem:cubical}
The scaled persistence modules 
\begin{itemize}
\item $\big(H(U_{2\alpha})\big)_{\alpha\ge 0}$ 
and the $L_\infty$-Rips  
module $\big(H(\rin_\alpha)\big)_{\alpha\ge 0}$
are $2$-approximations of each other, and 

\item $\big(H(U_{2\sqrt[4]{d}\alpha})\big)_{\alpha\ge 0}$ and the Rips 
module $\big(H(\ri_{\alpha})\big)_{\alpha\ge 0}$
$2d^{0.25}$-approximate each other.
\end{itemize}
\end{theorem}

To compute the cubical tower, we simply re-use the algorithm
for the simplicial case, with small changes:
\begin{itemize}
\item In the simplicial case, we used a container $S$ to hold the simplices
from the previous scale.
We alter $S$ to store the cubes from the previous scale.
For each interval, we store an id and its coordinates.
Each cube is stored as the set of ids of the intervals that define it.

\item At each scale, we enumerate the image of the cubical map by
computing the image of each interval, and then use this pre-computed
map to compute the image of $(\ge 1)$-dimensional cubes.

\item For the inclusions, we find all the active and secondary faces but do 
not compute the simplices.
The inclusions in the cubical tower correspond exactly to the inclusions
of active and secondary faces in the simplicial tower, 
so this enumerates all inclusions correctly.
\end{itemize}
From Lemma~\ref{lemma:num_activefaces} at most $n3^d$ active faces 
are added to the tower.
Hence at most $n3^d3^d=n6^d$ active and secondary faces are added to the tower.
Computing the tower takes time as in Theorem~\ref{theorem:algo1_bary}
by replacing $M$ with the size bound.
We conclude that:

\begin{theorem}
\label{theorem:cubical-comp}
The cubical tower has size at most $n6^{d}$ and takes at most $n6^{d}\log \Delta$ time 
in expectation to compute, where $\Delta$ is the spread of the point set.
\end{theorem}

\section{Discussion}
\label{section:conclusion}

\subsection{Practicality}
\label{subsec:prac}

We now touch upon the practical aspects of our constructions.
An implementation of our approximation scheme would be a tool that
computes the (approximate) persistence barcode for any input data set.
For any scheme to be useful in practice, it should be able to compute
sufficiently close approximations using a reasonable amount of resources.

Our cubical tower consists of cubical complexes connected via cubical maps.
To our knowledge, there are no algorithms to compute barcodes
in this setting where the cubical maps are more than just trivial inclusions.
As such, although our cubical scheme has exponentially lower theoretical
guarantees compared to the simplicial tower, we can not hope to test it
in practice unless the appropriate primitives are available.
It could be an interesting research direction to develop this primitive
and in particular investigate whether the techniques used in 
computing persistence barcodes for a simplicial tower allow a generalization
to the cubical case.

It makes more sense to inspect the simplicial tower.
We saw in Theorem~\ref{theorem:towersize} that the size of the tower is 
$n6^{d-1}(2k+4)(k+3)! \left\{\begin{array}{c}d\\k+2\end{array}\right\}$.
Unfortunately, this bound is already too large so that the storage
requirement of the Algorithm (Theorem~\ref{theorem:cubical-comp}) 
explodes exponentially.
Let us assume a conservative bound of 1 Byte of memory requirement per simplex.
For a point set in $d=8$ dimensions and $k=4$, the complexity
bound is already at least $4000$ Terabytes, before factoring in $n$.
For a point set in $d=10$ dimensions and $k=5$, this explodes 
to $10^{20}$ Terabytes.
While these are upper bounds, in practice the complexity will still need to be many
orders of magnitude smaller to be feasile, which is unlikely.
Even with conservative estimates our storage requirement is impractical.

Therefore we are not very hopeful that implementing the scheme in its
current state will provide any useful insight for high dimensional approximations.
Making it implementation-worthy demands more optimizations and tools at the algorithmic level.
This is worth another Algorithmic engineering project in its own right.
We plan to pursue this line of research in the future.
Since our focus in this paper was geared towards theoretical aspects
of approximations, we exclude experimental results in the current work.
We hope that a more careful implementation-focussed approach 
may prove more practical.

On the other hand, the upper bound for the cubical case is simply $n6^{d}$.
Even for $d=10$, the storage requirement would be less than $100$ Megabytes 
before factoring in $n$.
This is far more attractive than the simplicial case.
As such, it may make more sense to invest time and effort in developing tools
to compute barcodes in the cubical setup.

\subsection{Summary}
\label{subsec:outlook}

We presented an approximation scheme for the Rips filtration, 
with improved approximation ratio, size and computational complexity 
than previous approaches for the case of high-dimensional point clouds.
In particular, we are able to achieve a marked reduction in the size
of the approximation by using cubical complexes in place of simplicial
complexes.
This is in contrast to all other previous approaches that used simplicial
complexes as approximating structures.

An important technique that we used in our scheme is the application of 
acyclic carriers to prove interleaving results. 
An alternative would to be explicitly construct chain maps between
the Rips and the approximation towers; unfortunately, this make
the interleaving analysis significantly more complex.
While the proof of the interleaving in
Section~\ref{subsection:bary_interleaving} is still technically challenging,
it greatly simplifies by the usage of acyclic carriers.
There is also no benefit in knowing the interleaving maps 
because they are only required for the analysis of the interleaving, 
and not for the actual computation of the approximation tower.
We believe that this technique is of general interest for the construction
of approximations of cell complexes.

Our simplicial tower is connected by simplicial maps; there are 
(implemented) algorithms to compute the barcode of such 
towers~\cite{dfw-gic,ks-twr}.
It is also quite easy to adapt our tower construction to a streaming 
setting~\cite{ks-twr}, where the output list of events is passed to an 
output stream instead of being stored in memory.

\appendix

\section{Strong Interleaving for Barycentric scheme}
\label{section:appendix-missing}
\label{subsection:appendix-strong}

Recall that we build the approximation tower over the set of scales
$I:=\{\alpha_s = 2^s \mid s\in \Z\}$.
The tower $(\ux_{\alpha})_{\alpha\in I}$ connected with the simplicial
map $\tg$ can be extended to the set of scales $\{\alpha\ge 0\}$ with
simple modifications:
\begin{itemize}
\item for $\alpha\in I$, we define $\ux_\alpha$ in the usual manner.
The map $\tg$ stays the same as before for complexes at such scales.

\item for all $\alpha \in [\alpha_s,\alpha_{s+1}) $, we set 
$\ux_\alpha=\ux_{\alpha_s}$, for any $\alpha_s\in I$.
That means, the complex stays the same in the interval between any
two scales of $I$, so we define $\tg$ as the identity within this interval.
\end{itemize}
These give rise to the tower $(\ux_\alpha)_{\alpha\ge 0}$,
that is connected with the simplicial map $\tg$.
This modification helps in improving the interleaving with the 
Rips persistence module.

First, we extend the acyclic carriers $C_1$ and $C_2$ from before 
to the new case:
\begin{itemize}
\item $C_1^\alpha:\rin_{\alpha}\rightarrow \ux_{4\alpha},\alpha>0$:
we define $C_1$ as before, simply changing the scales in the definition.
It is straightforward to see that $C_1$ is still a well-defined acyclic carrier.

\item $C_2^\alpha:\ux_{\alpha} \rightarrow \rin_{\alpha},\alpha\ge 0$: this stays
the same as before. 
It is simple to check that $C_2$ is still a well-defined acyclic carrier.
\end{itemize}
These give rise to augmentation-preserving chain maps between 
the chain complexes:
\[
c_1^\alpha: \ch_\ast(\rin_{\alpha}) \rightarrow \ch_\ast(\ux_{4\alpha}) 
\qquad \text{and} \qquad 
c_2^\alpha: \ch_\ast(\ux_{\alpha}) \rightarrow \ch_\ast(\rin_{\alpha}),
\]
using the acyclic carrier theorem as before (Theorem~\ref{theorem:acyclic_carrier}).

\begin{lemma}
\label{lemma:bary_strong_1}
The diagram 
\begin{equation}
\label{equation:bary_strong_1}
\xymatrix{
	\ch_\ast(\rin_{\alpha}) \ar[r]^{inc} \ar[rd]^{c_1} 
	& \ch_\ast(\rin_{\alpha'}) \ar[rd]^{c_1}
	\\
	& \ch_\ast(\ux_{4\alpha}) \ar[r]^{\tg}
	& \ch_\ast(\ux_{4\alpha'}) 
}
\end{equation}
commutes on the homology level, for all $0\le \alpha\le \alpha'$.
\end{lemma}

\begin{proof}
Consider the acyclic carrier 
$C_1\circ inc:\rin_{\alpha} \rightarrow \ux_{4\alpha'} $.
It is simple to verify that this carrier carries both
$c_1\circ inc$ and $\tg \circ c_1$, so the induced diagram on the 
homology groups commutes, from Theorem~\ref{theorem:acyclic_carrier}.
\end{proof}

\begin{lemma}
\label{lemma:bary_strong_2}
The diagram 
\begin{equation}
\label{equation:bary_strong_2}
\xymatrix{
	&\ch_\ast(\rin_{\alpha}) \ar[r]^{inc} 
	& \ch_\ast(\rin_{\alpha'}) 
	\\
	\ch_\ast(\ux_{\alpha}) \ar[r]^{\tg} \ar[ru]^{c_2}
	& \ch_\ast(\ux_{\alpha'}) \ar[ru]^{c_2}
}
\end{equation}
commutes on the homology level, for all $0\le \alpha\le \alpha'$.
\end{lemma}

\begin{proof}
We construct an acyclic carrier $D: \ux_\alpha\rightarrow \rin_{\alpha'}$ 
which carries $inc\circ c_2$ and $c_2\circ \tg$, thereby proving
the claim (Theorem~\ref{theorem:acyclic_carrier}).

Consider any simplex $\sigma\in \ux_{\alpha}$ and let $E\in\square_{\alpha}$ 
be the minimal active face of containing $\sigma$.
We set $D(\sigma)$ as the simplex on the set of input points of $P$, 
which lie in the Voronoi regions of the vertices of $g(E)$.
By the triangle inequality, $D(\sigma)$ is a simplex of $\rin_{\alpha'}$,
so that $D$ is a well-defined acyclic carrier.
It is straightforward to verify that $D$ carries both $c_2\circ \tg$
and $inc\circ c_2$.
\end{proof}

\begin{lemma}
\label{lemma:bary_strong_3}
The diagram
\begin{equation}
\label{equation:bary_strong_3}
\xymatrix{
	&\ch_\ast(\rin_{\alpha}) \ar[r]^{inc} 
	& \ch_\ast(\rin_{\alpha'}) \ar[rd]^{c_1}
	\\
	\ch_\ast(\ux_{\alpha}) \ar[rrr]^{\tg} \ar[ru]^{c_2}
	& & & \ch_\ast(\ux_{4\alpha'}) 
}
\end{equation}
commutes on the homology level, for all $0\le \alpha\le \alpha'$.
\end{lemma}

\begin{proof}
The diagram is essentially the same as the lower triangle of 
Diagram~\ref{diagram:bary_inf_intlv}, with a change in the scales.
As a result, the proof of Lemma~\ref{lemma:bary_inf_intlv_low} also
applies for our claim directly.
\end{proof}

\begin{lemma}
\label{lemma:bary_strong_4}
The diagram 
\begin{equation}
\label{equation:bary_strong_4}
\xymatrix{
	\ch_\ast(\rin_{\alpha}) \ar[rrr]^{inc} \ar[rd]^{c_1} 
	& & &\ch_\ast(\rin_{4\alpha'}) 
	\\
	& \ch_\ast(\ux_{4\alpha}) \ar[r]^{\tg}
	& \ch_\ast(\ux_{4\alpha'}) \ar[ru]^{c_2}
}
\end{equation}
commutes on the homology level, for all $0\le \alpha\le \alpha'$.
\end{lemma}

\begin{proof}
The diagram can be re-interpreted as:
\begin{equation}
\xymatrix{
	\ch_\ast(\rin_{\alpha}) \ar[rr]^{inc} \ar[rd]^{\tg\circ c_1} 
	& & \ch_\ast(\rin_{4\alpha'}) 
	\\
	& \ch_\ast(\ux_{4\alpha'}) \ar[ru]^{c_2}
}
\end{equation}
The modified diagram is essentially the same as the upper triangle of 
Diagram~\ref{diagram:bary_inf_intlv}, with a change in the scales
and a replacement of $c_1$ with $\tg\circ c_1$, that is equivalent to
the chain map at the scale $\alpha'$.
Hence, the proof of Lemma~\ref{lemma:bary_inf_intlv_up} 
also applies for our claim directly.
\end{proof}

Using Lemmas~\ref{lemma:bary_strong_1}, \ref{lemma:bary_strong_2},
\ref{lemma:bary_strong_3}, \ref{lemma:bary_strong_4}, and the scale
balancing technique for strongly interleaved persistence modules, 
it follows that
\begin{lemma}
\label{lemma:bary_strong_full}
The persistence modules $\big(H(\ux_{2\alpha})\big)_{\alpha\ge 0}$ and 
$\big(H(\rin_\alpha)\big)_{\alpha\ge 0}$ are strongly $2$-interleaved.
\end{lemma}

\end{document}